%% file: BiWeitex.tex
\documentclass[aps,11pt]{revtex4}
\usepackage{algorithm}
\usepackage{makeidx}
\usepackage{multirow}
\usepackage{array}
\usepackage{amsmath,amssymb,amsthm}
\usepackage{mathrsfs}
\usepackage{amsfonts}
\usepackage{setspace}
\usepackage[]{xcolor}
\usepackage{graphicx,color,scalefnt}
\usepackage{epstopdf,mathtools,amsfonts}
\usepackage{ulem}
\usepackage{hyperref}
\usepackage{subfigure,epstopdf}
\newcommand{\be}{\begin{equation}}
\newcommand{\ee}{\end{equation}}
\newcommand{\bea}{\begin{eqnarray}}
\newcommand{\eea}{\end{eqnarray}}
\newtheorem{theorem}{Theorem}[section]

\newtheorem{corollary}[theorem]{Corollary}
\newtheorem{proposition}[theorem]{Proposition}
\newtheorem{remark}[theorem]{Remark}

\newcommand{\ud}{\mathrm{d}}

\begin{document}
\title{\Large \bf A Bimodal Weibull Distribution: Properties and Inference}

\smallskip 
\author{\bf Roberto Vila$^{a1}$}
\author{Mehmet Niyazi \c{C}ankaya$^{a2,a3}$}
\bigskip 
\affiliation{$^{a1}$ Departamento de Estatística, Universidade de Brasília, Brazil}

\affiliation{$^{a2}$ Faculty of Applied Sciences, Department of International Trading and Finance, U\c{s}ak University,}
\affiliation{$^{a3}$  Faculty of Art and Sciences, Department of Statistics, U\c{s}ak University, U\c{s}ak, Turkey}
\begin{abstract}
\begin{center}
	\text{\normalsize{Abstract}}
\end{center}
\vspace{-0,3cm}
 Modeling is a challenging topic and using parametric models is an important stage to reach flexible function for modeling.  Weibull distribution has two parameters which are shape $\alpha$ and scale $\beta$. In this study, bimodality parameter is added and so  bimodal Weibull distribution is proposed by using a quadratic transformation technique used to generate bimodal functions produced due to using the quadratic expression. The analytical simplicity of Weibull and quadratic form give an advantage to derive a bimodal Weibull via constructing normalizing constant. The characteristics and properties of the proposed distribution are examined to show its usability in modeling.  After examination as first stage in modeling issue, it is appropriate to use bimodal Weibull for modeling data sets. Two estimation methods which are maximum $\log_q$ likelihood and its special form including objective functions $\log_q(f)$ and $\log(f)$ are used to estimate the parameters of shape, scale and bimodality parameters of the function.  The second stage in modeling is overcome by using heuristic algorithm for optimization of function according to parameters due to fact that converging to global point of objective function is performed by heuristic algorithm based on the stochastic optimization. Real data sets are provided to show the modeling competence of the proposed distribution.
 
\end{abstract}
\nopagebreak


\maketitle
\section{Introduction}
\label{sec1}

Weibull distribution is very popular and used extensively in the applied field of science. There are two main parameters which are shape $\alpha$ and scale $\beta$ in probability density (p.d.) function of Weibull$(\boldsymbol{\theta})$, where $\boldsymbol{\theta}=(\alpha,\beta)$. The special values of shape and scale parameters drop to exponential and Rayleigh distributions \cite{Weibull}. Weibull distribution is used extensively and different forms of Weibull were generated by \cite{EPdist,Caste} and references therein. It has many applications in survival analysis, bioscience \cite{Martinez,OrtegaBetaWei}, management of maintenance, reliability engineering \cite{reliability}. 
The exponential kernels in Weibull and Gamma distributions are $\exp(-z^{\alpha})$ and $\exp(-z)$, respectively, which shows that Weibull is light tailed function if $\alpha>1$. Thus, Weibull can have advantegous when compared with Gamma for the light tailed empirical distributions. In order to have the different forms of tails at the positive axis on the real line, $q$-deformation can be used as an alternative solution to produce heavy-tailed functions \cite{Bercher12a}. It should be noted that $\log_q(f)={f^{1-q}-1 \over 1-q}$, $f \in[0,1]$, $q \in \mathbb{R} \backslash \{1\}$ as an objective function produces heavy-tailed function.   It is not easy to know which parametric model will be true one to represent accurately a population in the questionnaire, which makes  a challenge task for us to choose the true parametric model for the population. This is an outstanding topic in where it is difficult to finalize the dispute about choosing the true model. If p.d. function is defined on the real line, then location can be used to derive the bimodality around location \cite{Cankaya2018} and references therein. In order to derive a bimodal function defined on the positive axis of the real line, two p.d. functions can be tried to mix. In this case, the theoretical property of such mixed function should be examined and must be provided for existence of moments, entropies, tail behaviour, etc in order to apply it into modelling issue. Further, the computational complexity can arise due to the increased number of parameters which will be responsible to model subgroups assumed to be a member of subpopulation at a phenomena in questionnaire \cite{yangcomp}.  The managing bimodaliy for distribution defined on the positive axis is an another challenging task for us if data set shows a bimodality. Another task is about the applying a p.d. function which has important properties such as existence of moments, entropies, etc for conducting an accurate modelling on the data sets. The quadratic transformation for  normal distribution   is proposed by \cite{bimodal} to  produce bimodality and also asymmetry. Using the parameters whose role is defined exactly can guarantee the smoothness property (absolutely continuous function which is differantiable according to Radon Nikodym derivative \cite{knapp}) of the derived function as well.   In order to derive a bimodal distribution from a unimodal distribution, we can use quadratic transformation technique which gives a polynomial movement owing the fact that the used quadratic function can make a fluctation on the unimodal function. We will propose a new bimodal distribution generated by quadratic transformation technique with bimodality parameter $\delta$ and unimodal Weibull distribution with parameters shape $\alpha$ and scale $\beta$.  Thus, we can observe how the bimodality can be derived. After applying this technique, we can get a normalizing constant to produce a p.d. function.  The property of a newly generated function should be extensively examined in order to pass the test about finiteness of moments, existence of entropies, etc. Thus, the proposed distribution can be used to model a data set.  

It can be generally observed that a phenomena in the universe can show a bimodality. There can be a contamination or an irregularity into underlying or majority of the data. If the replicated values at an interval on the real line are observed, then bimodality occurs. In other words, a data set can be a combination of different parametric models or different values of same parametric model $f(x;\boldsymbol{\theta})$  to represent bimodality \cite{mix1,mix2,yangcomp}. The light tailed property, tractability of analytical expression, bimodality kernel $1+(1-\delta x)^2$ used to derive only one mode property occurred due to the polynomial degree $2$ \cite{bimodal}, and also for confining yourself for modelling data set having two modes at the different degree high of peakedness can require to apply bimodality generator for Weibull distribution.

The organization of this paper is as following forms. Bimodal Weibull distribution is proposed by Section \ref{bimodalWei}. Several properties of the proposed parametric model are examined and introduced by Sections  \ref{properties} and \ref{characte}. The section \ref{inference} introduces that the model parameters are estimated by maximum $\log$ and $\log_q$ likelihood methods. Section \ref{appsection} gives numerial assessments of the proposed model and also provides the comparison between recently proposed Bimodal Gamma (BGamma) distribution \cite{Vila2020} and the model in this paper. The conclusion is given by \ref{concsect}.

\section{Bimodal Weibull distribution}\label{bimodalWei}
We say that a non-negative random variable $X$ has a bimodal Weibull distribution with vector parameter $\boldsymbol{\theta}=(\alpha,\beta,\delta)$, denoted by 
$X\sim \text{BWeibull}(\boldsymbol{\theta})$, if its p.d. function is given by
\begin{align}\label{def-Weibull}
f(x;\boldsymbol{\theta})
=
{\alpha\over \beta Z_{\boldsymbol{\theta}}}\, 
\big[1+(1-\delta x)^2\big]\, \bigg({x\over\beta}\bigg)^{\alpha-1}
\exp\Bigg[-\bigg({x\over\beta}\bigg)^{\alpha}\Bigg],
\quad x\geqslant 0; \ \alpha,\beta>0,\ \delta\in\mathbb{R},
\end{align}
where $Z_{\boldsymbol{\theta}}$ is the normalization constant given by
\begin{align}\label{partition}
Z_{\boldsymbol{\theta}}=2+
\delta^2\beta^2\Gamma\bigg(1+{2\over\alpha}\bigg)
-
2\delta\beta\Gamma\bigg(1+{1\over\alpha}\bigg),
\end{align}
which depends only on the parameter vector $\boldsymbol{\theta}$, $\alpha$ is the shape parameter, $\beta$ is the scale parameter and $\delta$ is the parameter that controls the  uni- or bimodality of the distribution. Here, $\Gamma(z)$ denotes the complete gamma function.

The behavior of $f(x;\boldsymbol{\theta})$ with $x\to 0$ or $x\to \infty$ is as follows:
\begin{align*}
\lim_{x\to 0} f(x;\boldsymbol{\theta})
=
\begin{cases}
\infty & \text{for} \quad 0<\alpha<1,
\\
\dfrac{2}{\beta Z_{\boldsymbol{\theta}}} & \text{for} \quad \alpha=1,
\\
0 & \text{for} \quad \alpha>1,
\end{cases}
\end{align*}
\begin{align*}
\lim_{x\to \infty} f(x;\boldsymbol{\theta})
=
0 \quad \forall \alpha>0.
\end{align*}
It is verified that the cumulated distribution (c.d.) function of $X\sim \text{BWeibull}(\boldsymbol{\theta})$ is given by (see Item 2 of Proposition \ref{prop-3} for more details)
\begin{align*}
F(x;\boldsymbol{\theta})
=
\dfrac{\displaystyle
	2-\big[1+(1-\delta x)^2\big]\,
	\exp\Bigg[-\bigg({x\over\beta}\bigg)^\alpha\Bigg]
	-
	{2\delta\beta\over\alpha}\,
	\bigg[ \gamma\bigg({1\over\alpha},{x^\alpha\over\beta^\alpha}\bigg)
	-
	{\delta\beta}\, \gamma\bigg({2\over\alpha},{x^\alpha\over\beta^\alpha}\bigg)
	\bigg]
}{\displaystyle
	2+
	\delta^2\beta^2\Gamma\bigg(1+{2\over\alpha}\bigg)
	-
	2\delta\beta\Gamma\bigg(1+{1\over\alpha}\bigg)},
\quad x\geqslant 0,
\end{align*}
where $\Gamma(s,x)$, $s>0$, is the upper incomplete gamma function, and $\gamma(s,x)$, $s>0$, is the lower incomplete gamma function.
Related quantities such as reliability, hazard rate and the mean residual life can be found in Subsection
\ref{Reliability-hazard rate-mean residual life}.

The p.d and c.d. functions (PDF and CDF) in Figures \ref{figpdfs} and \ref{figcdfs}  are drawn for different values of parameters.
\begin{figure}[htbp]
	\centering
	\subfigure[PDF of BWeibull for $\delta$]{\label{fig:pdf11}\includegraphics[width=0.4\textwidth]{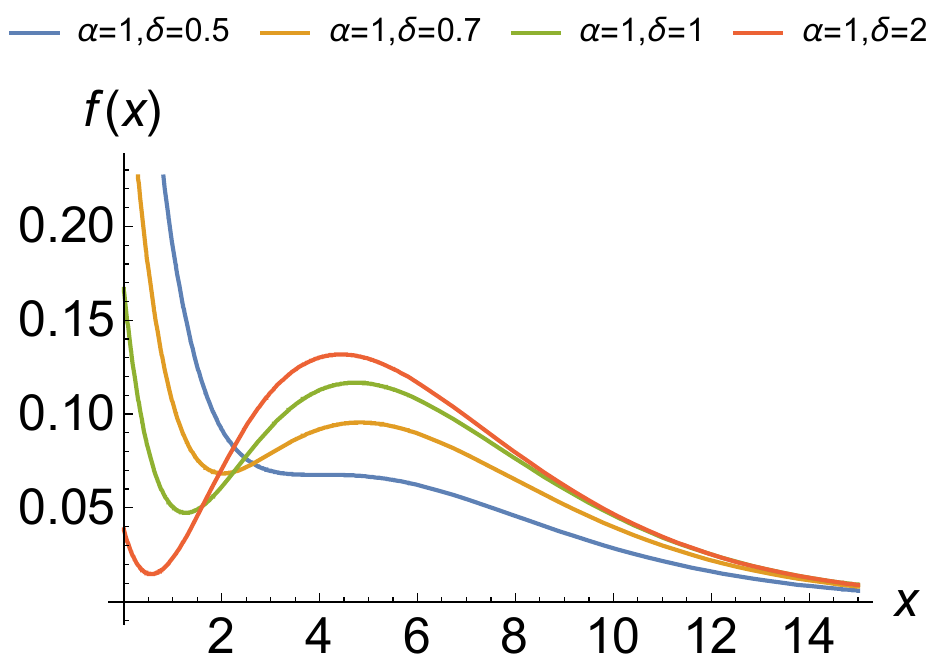}}
	\subfigure[PDF of BWeibull for  $\alpha$ and $\delta$]{\label{fig:pdf2}\includegraphics[width=0.4\textwidth]{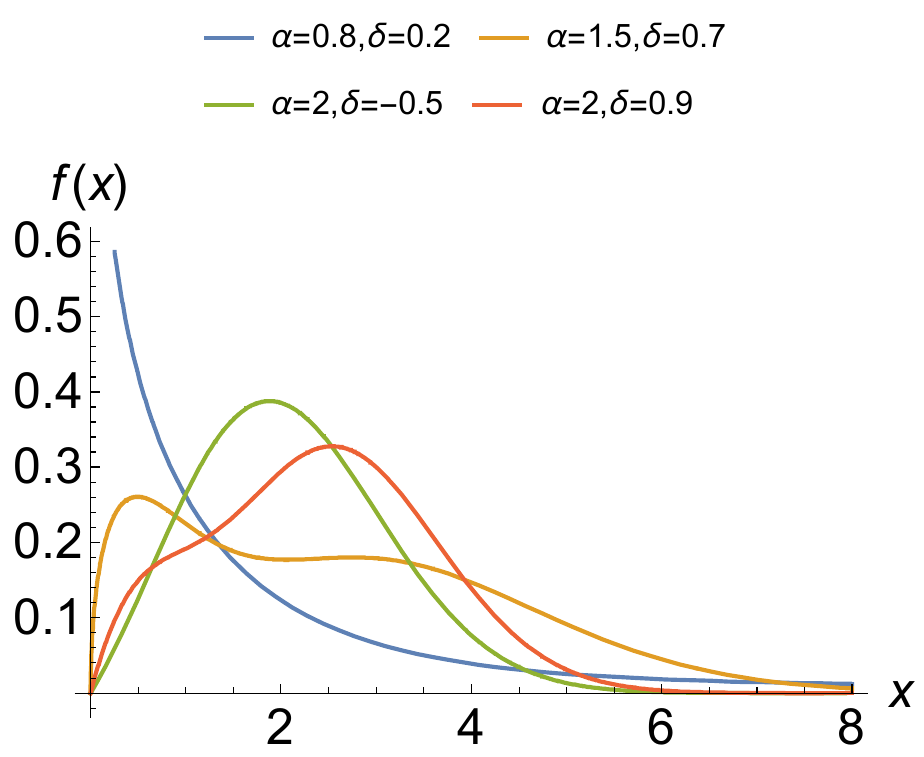}}
	\caption{PDF of BWeibull with different values of parameters $\alpha,\delta$ and $\beta=2$.}
	\label{figpdfs}
\end{figure}
\begin{figure}[htbp]
	\centering
	\subfigure[CDF of BWeibull for $\delta$]{\label{fig:cdf1}\includegraphics[width=0.4\textwidth]{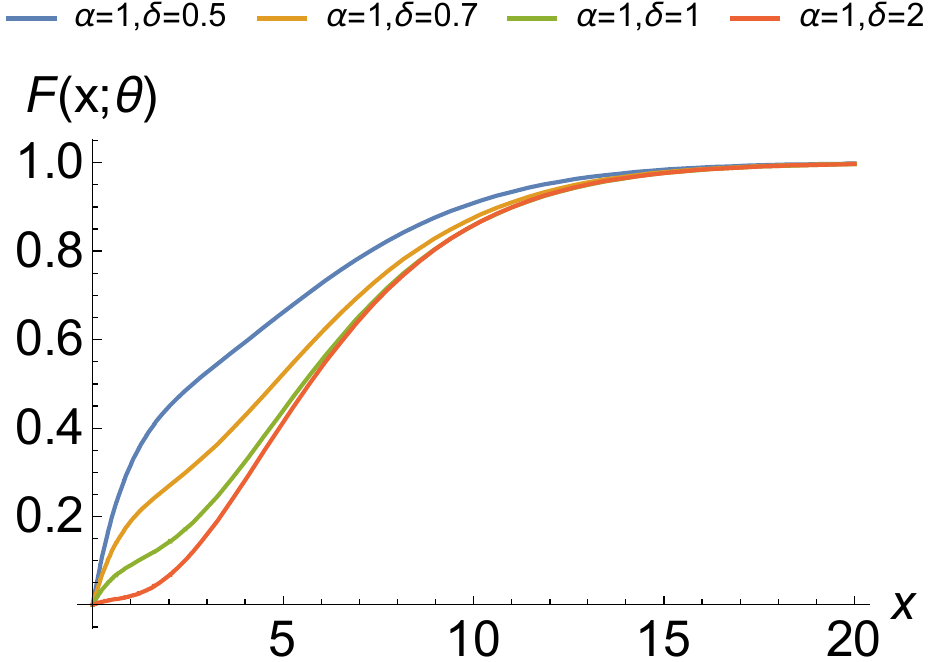}}
	\subfigure[CDF of BWeibull for  $\alpha$ and $\delta$]{\label{fig:cdf2}\includegraphics[width=0.4\textwidth]{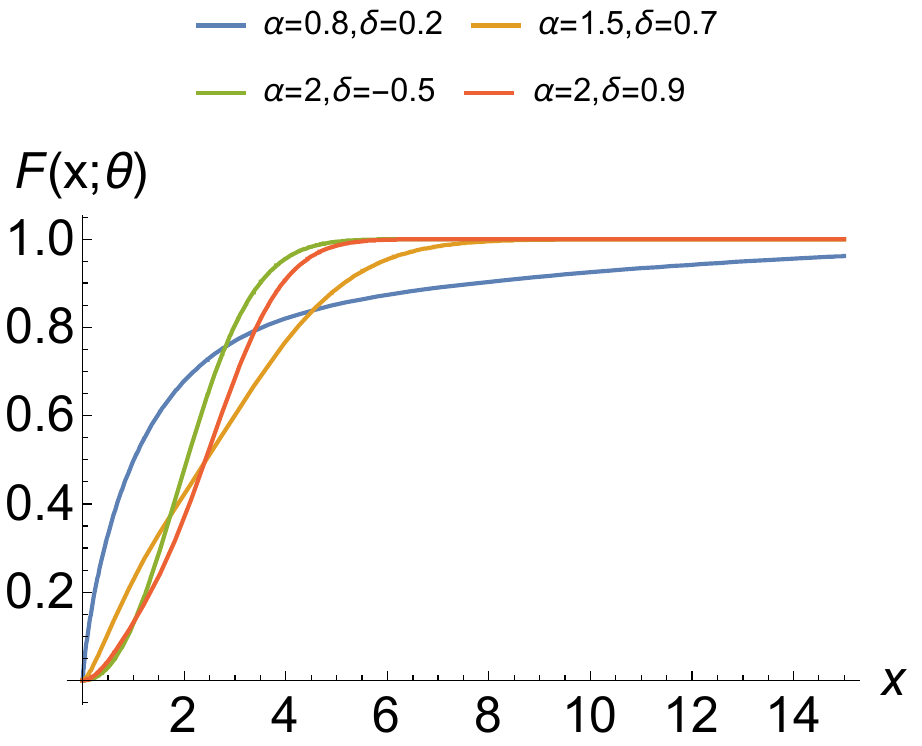}}
	\caption{CDF of BWeibull with different values of parameters $\alpha,\delta$ and $\beta=2$.}
	\label{figcdfs}
\end{figure}



\section{Unimodality and Bimodality properties}\label{properties}

\begin{proposition}\label{mode}
A point $x$ is a mode of the \text{BWeibull} density \eqref{def-Weibull}, if
and only if it is a solution of the following equation
\begin{align}\label{eq-mode}
\alpha\delta^2 x^{\alpha+2}
-
2\alpha\delta x^{\alpha+1}
+
2\alpha x^{\alpha}
-
(\alpha+1)\beta^\alpha\delta^2 x^2
+
2\alpha\beta^\alpha\delta x
-
2(\alpha-1)\beta^\alpha=0.
\end{align}
\end{proposition}
\begin{proof}
The proof is trivial and omitted.
\end{proof}

\begin{proposition}\label{prop-1}
If $\alpha> 2$ is a natural number, the following it hold:
	\begin{itemize}
		\item[1)] If $\delta<0$ then the equation \eqref{eq-mode} has an unique positive root;
		\item[2)] If $\delta> 0$ then the equation \eqref{eq-mode} has five, or three, or one positive roots.
	\end{itemize}
\end{proposition}
\begin{proof}
	The proof follows directly by applying Descartes' rule of signs (see, e.g. Xue 2012 \cite{Xue2012}).
\end{proof}

\begin{proposition}[Unimodality and monotonicity]
Let $X\sim \text{BWeibull}(\boldsymbol{\theta})$ with $\alpha=1$. The following hold:
\begin{itemize}
\item[1)] If $\delta< -1/\beta$ then  the \text{BWeibull} density is unimodal;
\item[2)] If $\delta>0$ then  the \text{BWeibull} density is unimodal or strictly decreasing.
\end{itemize}
\end{proposition}
\begin{proof}
Let	 $\alpha=1$. In this case note that the equation \eqref{eq-mode} can be written as
\begin{align*}
p_2(x)=\delta^2x^2-2\delta(1+\beta\delta)x+2(1+\beta\delta)=0.
\end{align*}

Firstly we assume that $\delta< -1/\beta$. By Descartes' rule of signs the above equation has a unique positive root, denoted by $x_0$. Since 
$\lim_{x\to 0} f(x;\boldsymbol{\theta})=2 /(\beta Z_{\boldsymbol{\theta}}) $ and $\lim_{x\to \infty} f(x;\boldsymbol{\theta})=0$ we have that $x_0$ is the unique maximum point with maximum value $f(x_0;\boldsymbol{\theta})>2/(\beta Z_{\boldsymbol{\theta}})$. This proves the statement in the first item. 

In the second case, $\delta>0$, Descartes' rule of signs guarantees that the equation $p_2(x)=0$ has two roots, or zero roots. Since
$\lim_{x\to 0} f(x;\boldsymbol{\theta})=2/(\beta Z_{\boldsymbol{\theta}}) $ and $\lim_{x\to \infty} f(x;\boldsymbol{\theta})=0$, the proof of second item follows.
\end{proof}

Let us define
\begin{align}\label{discriminant}
\Delta&=   256 a^{3}e^{3} - 192 a^{2}bde^{2} - 128 a^{2}c^{2}e^{2} + 144 a^{2}cd^{2}e -27 a^{2}d^{4}
\nonumber
\\&
+144 ab^{2}ce^{2} - 6 ab^ {2}d^{2}e - 80 abc^{2}de + 18 abcd^{3} + 16 ac^{4}e 
\\&
-4 ac^{3}d^{2} - 27 b^{4}e^{2} + 18 b^{3}cde - 4 b^{3}d^{3} - 4 b^{2}c^{3}e + b^{2}c^{2}d^{2}
\nonumber
\end{align}
the discriminant of the quartic polynomial $ax^{4}+bx^{3}+cx^{2}+dx+e$, where
\begin{align}\label{coefficients}
a=2\delta^2,  \quad 
b=-4\delta , \quad
c=4-3\beta^2\delta^2, \quad
d=4\beta^2\delta, \quad
e=-2\beta^2.
\end{align} 
\begin{theorem}[Bimodality and unimodality]
Let $X\sim \text{BWeibull}(\boldsymbol{\theta})$ with $\alpha=2$ and $\delta>0$. The following hold:
\begin{itemize}
	\item[1)] If $\Delta>0$ and $13/12<\beta^2\delta^2<4/3$, then the \text{BWeibull} density is bimodal;
	\item[2)] If $\Delta< 0$ and $\beta^2\delta^2\leqslant 4/{3}$, then the \text{BWeibull} density is unimodal;
	\item[3)] If $\Delta= 0$ and $\beta^2\delta^2=4/{3}$, then the \text{BWeibull} density is unimodal;
\end{itemize}
where $\Delta$ is defined by \eqref{discriminant} and \eqref{coefficients}.
\end{theorem}
\begin{proof}
Let $x$ be a mode of the \text{BWeibull} density. By Proposition \ref{mode} this mode satisfies the equation \eqref{eq-mode}, with $\alpha=2$, or equivalently this one is solution of the quartic equation 
\begin{align*}
p_4(x)=ax^{4}+bx^{3}+cx^{2}+dx+e=0
\end{align*}
with real coefficients defined in \eqref{coefficients}. By using that $\delta>0$, in  both cases, $\beta^2\delta^2<4/{3}$ or $\beta^2\delta^2=4/{3}$ (in this case, $c=0$), Descartes' rule of signs guarantees that:
\begin{align}\label{aff}
\text{The polynomial} \ p_4(x) \ \text{has three, or one positive roots and one negative root.}
\end{align}

To check Items 1), 2) and 3), we first define the following quantities:
\begin{align*}
&P=8ac-3b^2
=
4\delta^2(13-12\beta^2\delta^2);
\\
&D=64a^3e-16a^2c^2+16ab^2c-16a^2bd-3b^4
=
64(4-9\beta^4\delta^4);
\\
&\Delta_0=c^2-3bd+12ae
=16+9\beta^4\delta^4-24\beta^2\delta^2.
\end{align*}

1)
If $\Delta>0$ then either the equation's four roots $p_4(x)=0$ are all real or none is.
Since
$P < 0$ and $D < 0$ whenever $\beta^2\delta^2>13/12$ and $\beta^2\delta^2>2/3$, respectively,  then all four roots of $p_4(x)$ are real and distinct. Then, from affirmation \eqref{aff} it follows that the equation $p_4(x)=0$ has three distinct positive roots, denoted by $x_1, x_2, x_3$, and one negative root. Let's assume that $x_1<x_2<x_3$. Since
$\lim_{x\to 0} f(x;\boldsymbol{\theta})= \lim_{x\to \infty} f(x;\boldsymbol{\theta})=0$ we have that $x_1$ and $x_3$ are two
maximum points and $x_1$ is the unique minimum point. This proves the bimodality property in the first item.

\smallskip
2)
If $\Delta<0$ then the equation  $p_4(x)=0$ has two distinct real roots and two complex conjugate non-real roots. Then, from \eqref{aff} it follows that the equation $p_4(x)=0$ has one positive root, denoted by $x_0$, and one negative root. Since
$\lim_{x\to 0} f(x;\boldsymbol{\theta})= \lim_{x\to \infty} f(x;\boldsymbol{\theta})=0$, it follows that  $x_0$ is the unique maximum point. This proves the second item.

\smallskip
3)
If $\Delta=0$ then the polynomial  $p_4(x)$ has a multiple root. 
Since $\Delta_0 = 0$ and $D=-768 \neq 0$ when $\beta^2\delta^2=4/{3}$, there are a triple root and a simple root, all real. Then, by affirmation \eqref{aff},  the equation $p_4(x)=0$ has a positive triple root, denoted by $x_0$, and one negative root.
Since
$\lim_{x\to 0} f(x;\boldsymbol{\theta})= \lim_{x\to \infty} f(x;\boldsymbol{\theta})=0$ we have that $x_0$ is the unique maximum point. This completes the proof of third item.
\end{proof}

\nopagebreak 
\section{Characteristics of Probability Denstiy Function}\label{characte}
\subsection{Real moments}
\begin{proposition} \label{moments}
	If $X\sim \text{BWeibull}(\boldsymbol{\theta})$ then
\begin{align*}
\mathbb{E}(X^r)
=
\beta^r\,
\dfrac
{\displaystyle 
2\Gamma\bigg(1+{r\over\alpha}\bigg)+\delta^2\beta^2\Gamma\bigg(1+{r+2\over\alpha}\bigg)
-
2\delta\beta\Gamma\bigg(1+{r+1\over\alpha}\bigg)}
{\displaystyle 
2+\delta^2\beta^2\Gamma\bigg(1+{2\over\alpha}\bigg)
-
2\delta\beta\Gamma\bigg(1+{1\over\alpha}\bigg)}, \quad r>-\alpha,
\end{align*}
where $\Gamma(z)$ is the complete gamma function.
\end{proposition}
\begin{proof}
A simple observation shows that
\begin{align*}
\mathbb{E}(X^r)
=
{1\over Z_{\boldsymbol{\theta}}}\, 
\left[2\mathbb{E}(Y^r)+\delta^2\mathbb{E}(Y^{r+2})-2\delta\mathbb{E}(Y^{r+1})\right],
\end{align*}
where $Y\sim \text{Weibull}(\alpha,\beta)$ and $Z_{\boldsymbol{\theta}}$ is as in \eqref{partition}. 
By using the following formula of Item (6) in  \cite{Cankaya2018}: 
\begin{align*}
\Gamma\bigg(s+{1\over\alpha}\bigg)
=
\alpha p^{\alpha s+1} 
\int_{0}^{\infty} 
y^{\alpha s}\, \exp\big[-(yp)^\alpha\big]
\, {\rm d}y,
\end{align*}
we have 
\begin{align*}
\mathbb{E}(Y^r)
=
\beta^r \Gamma\bigg(1+{r\over\alpha}\bigg), \quad r>-\alpha.
\end{align*}
By combining the above identities, the proof follows.
\end{proof}

As a consequence of the above proposition, the closed expressions for the moments, variance, skewness and kurtosis of random variable $X$ are easily obtained. The Figures \ref{fig:sk} and \ref{fig:ku} give the skewness and kurtosis for the different values of bimodality parameter $\delta$, respectively.  For the tried values of $\alpha$ and $\beta$, we observe that when the value of  $\delta$  are not large, we get a large scale of skewness and kurtosis. The Figure \ref{figskekur} represents that the parameter $\delta$ increases the flexiblity of function.
\begin{figure}[htb!]
	\centering
	\subfigure[Skewness of BWeibull for $\alpha$]{\label{fig:sk}\includegraphics[width=0.4\textwidth]{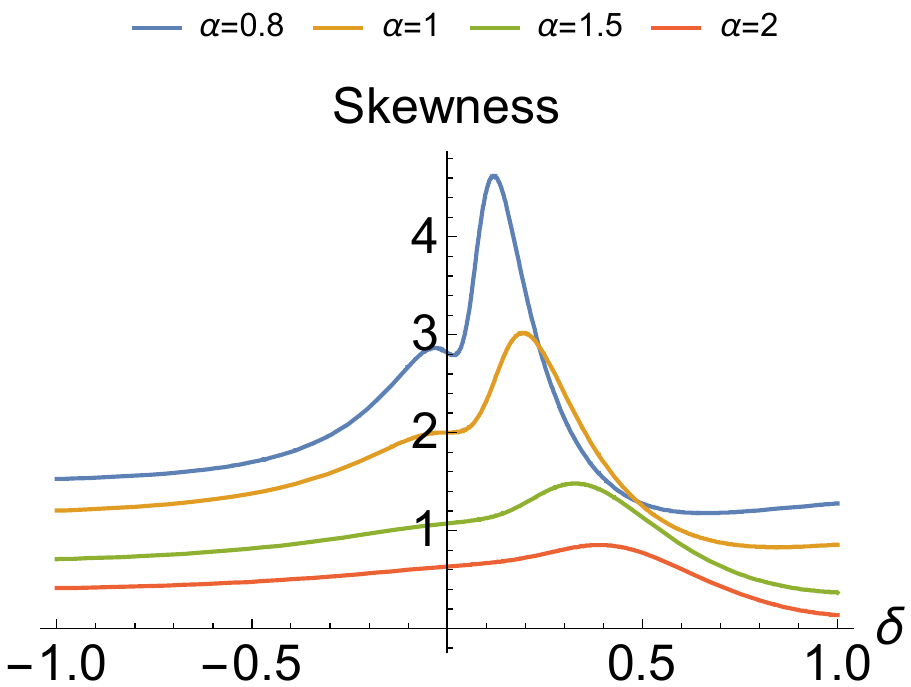}}
	\subfigure[Kurtosis of BWeibull for  $\alpha$]{\label{fig:ku}\includegraphics[width=0.4\textwidth]{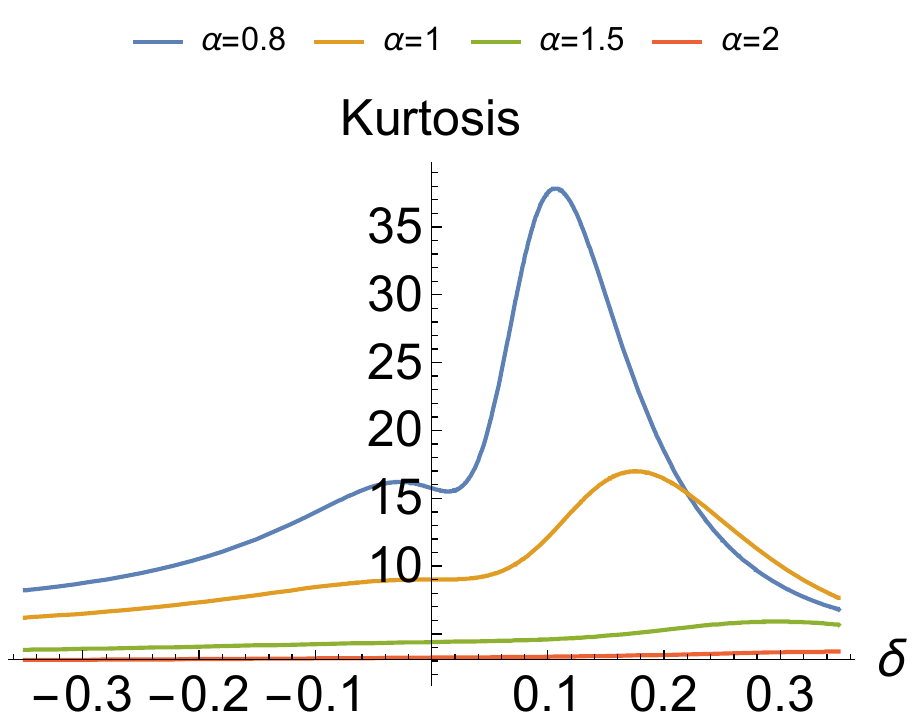}}
	\caption{Skewness and kurtosis of BWeibull with different values of $\alpha$ and $\beta=2$.}
	\label{figskekur}
\end{figure}

\subsection{Moment generating function}
\begin{theorem}\label{MGF}
	If $X\sim\text{BWeibull}(\boldsymbol{\theta})$ then the moment generating function $M_X(t)=\mathbb{E}[\exp({tX})]$ can be expressed as
{\scalefont{0.955}
\begin{align*}
M_X(t)=
\begin{cases}
\displaystyle
\sum_{n=0}^{\infty} {(\beta t)^n}\, 
\dfrac{1+\dfrac{1}{2}\,  {\delta^2\beta^2}(n+2)(n+1)-\delta\beta(n+1)}{1+\delta^2\beta^2-\delta\beta}& \text{for} \ \alpha=1 \ \text{and} \ \vert t\vert<1/\beta,
\\[0,4cm]
\displaystyle
\sum_{n=0}^{\infty} {(\beta t)^n\over n!}\,
\dfrac
{\displaystyle 
	2\Gamma\bigg(1+{n\over\alpha}\bigg)+\delta^2\beta^2\Gamma\bigg(1+{n+2\over\alpha}\bigg)
	-
	2\delta\beta\Gamma\bigg(1+{n+1\over\alpha}\bigg)}
{\displaystyle 
	2+\delta^2\beta^2\Gamma\bigg(1+{2\over\alpha}\bigg)
	-
	2\delta\beta\Gamma\bigg(1+{1\over\alpha}\bigg)}
& \text{for} \ \alpha>1 \ \text{and} \ t\in\mathbb{R}.
\end{cases}
\end{align*}
} 
\end{theorem}
\begin{proof}
Using the series expansion of $\exp({tX})$ we have
\begin{align*}
	M_X(t)=\mathbb{E}\Bigg[\sum_{n=0}^{\infty} {(tX)^n\over n}\Bigg].
\end{align*}
Assuming the validity of the following identity
\begin{align}\label{objective}
\mathbb{E}\Bigg[\sum_{n=0}^{\infty} {(tX)^n\over n}\Bigg]
=
\sum_{n=0}^{\infty} {t^n\over n!}\, \mathbb{E}(X^n),
\quad \alpha\geqslant 1,
\end{align}
note that by using Proposition \ref{moments} and the identity $\Gamma(n)=(n-1)!$, for $n=1,2,3,\ldots$ (in the case $\alpha=1$); the proof of the proposition follows.

\smallskip 
In what remains of the proof we prove the identity \eqref{objective}.
Indeed, from Monotone Convergence Theorem \cite{knapp} we have
\begin{align}\label{first-relation}
\mathbb{E}\left[
\sum_{n=0}^{\infty} \bigg\vert{(tX)^n\over n}\bigg\vert\right]
=
\sum_{n=0}^{\infty} \mathbb{E}\Bigg[\bigg\vert{(tX)^n\over n}\bigg\vert\Bigg]
=
\sum_{n=0}^{\infty} 
{\vert t\vert^n\over n!}\, \mathbb{E}(X^n),
\end{align}
where, by Proposition \ref{moments}, the series $\sum_{n=0}^{\infty} 
{\vert t\vert^n\over n!}\, \mathbb{E}(X^n)$ is written as
\begin{align}\label{cond-series}
\sum_{n=0}^{\infty} {\vert t\vert^n\over n!}\, \mathbb{E}(X^n)
=
\sum_{n=0}^{\infty} {(\beta\vert t\vert)^n\over n!}\,
\dfrac
{\displaystyle 
	2\Gamma\bigg(1+{n\over\alpha}\bigg)+\delta^2\beta^2\Gamma\bigg(1+{n+2\over\alpha}\bigg)
	-
	2\delta\beta\Gamma\bigg(1+{n+1\over\alpha}\bigg)}
{\displaystyle 
	2+\delta^2\beta^2\Gamma\bigg(1+{2\over\alpha}\bigg)
	-
	2\delta\beta\Gamma\bigg(1+{1\over\alpha}\bigg)}.
\end{align}

Note that the series $u(k)$ defined by
\begin{align*}
u(k)=
\sum_{n=0}^{\infty} a_n
=
\sum_{n=0}^{\infty} 
{(\beta\vert t\vert)^n\over n!}\,
\Gamma\bigg(1+{n+k\over\alpha}\bigg),
\quad	\quad k=0,1,2,
\end{align*}
converges when $\alpha\geqslant 1$, because by the ratio test the limit 
\begin{align*}
L&=\lim_{n\to\infty}\left\vert {a_{n+1}\over a_n}\right\vert
\\
&=
{\beta\vert t\vert}
\lim_{n\to\infty}
\bigg({1\over n+k}\bigg)\, \bigg(1+{k\over n+1}\bigg)\, 
{\displaystyle \Gamma\bigg({n+1+k\over\alpha}\bigg)\over\displaystyle \Gamma\bigg({n+k\over\alpha}\bigg)},
\quad k=0,1,2,
\\
&=
\begin{cases}
0 & \text{for} \quad \alpha>1,
\\
{\beta \vert t\vert} & \text{for} \quad \alpha=1,
\\
\infty & \text{for} \quad 0<\alpha<1,
\end{cases}
\end{align*}
is less than 1  when $\alpha>1$ $\forall t\in\mathbb{R}$; and $\alpha=1$ $\forall \vert t\vert<1/\beta$. 

Therefore, the series \eqref{cond-series} 
\begin{align*}
&
\sum_{n=0}^{\infty} {\vert t\vert^n\over n!}\, \mathbb{E}(X^n)
\\[0,2cm]
&=
\dfrac
{\displaystyle 
	2\sum_{n=0}^{\infty} {(\beta\vert t\vert)^n\over n!}\,
	\Gamma\bigg(1+{n\over\alpha}\bigg)
	+
	\delta^2\beta^2
	\sum_{n=0}^{\infty} {(\beta\vert t\vert)^n\over n!}\,
	\Gamma\bigg(1+{n+2\over\alpha}\bigg)
	-
	2\delta\beta
	\sum_{n=0}^{\infty} {(\beta\vert t\vert)^n\over n!}\,
	\Gamma\bigg(1+{n+1\over\alpha}\bigg)}
{\displaystyle 
	2+\delta^2\beta^2\Gamma\bigg(1+{2\over\alpha}\bigg)
	-
	2\delta\beta\Gamma\bigg(1+{1\over\alpha}\bigg)}
\\[0,2cm]
&=
\dfrac
{\displaystyle 
	2u(0)+\delta^2\beta^2u(2)
	-
	2\delta\beta u(1)}
{\displaystyle 
	2+\delta^2\beta^2\Gamma\bigg(1+{2\over\alpha}\bigg)
	-
	2\delta\beta\Gamma\bigg(1+{1\over\alpha}\bigg)}
\end{align*}
converges for $\alpha\geqslant 1$, and then
\begin{align*}
\mathbb{E}\left[
\sum_{n=0}^{\infty} \bigg\vert{(tX)^n\over n}\bigg\vert\right]
\stackrel{\eqref{first-relation}}{=}
\sum_{n=0}^{\infty} {\vert t\vert^n\over n!}\, \mathbb{E}(X^n)<\infty,
\end{align*}
when $\alpha\geqslant 1$.
So, applying Fubini's Theorem \cite{knapp} we get
\begin{align*}
M_X(t)=\mathbb{E}\Bigg[\sum_{n=0}^{\infty} {(tX)^n\over n}\Bigg]
=
\sum_{n=0}^{\infty} {t^n\over n!}\, \mathbb{E}(X^n),
\end{align*}
whenever $\alpha\geqslant 1$.
This proves \eqref{objective}.
We thus complete the proof.
\end{proof}

\begin{corollary}[Light-tailed distribution] If $X\sim\text{BWeibull}(\boldsymbol{\theta})$ and $\alpha\geqslant 1$, then there exists $t_0 > 0$ such that $\mathbb{P}(X > x)\leqslant \exp(-t_0x)$ for $x$ large enough.

\end{corollary}
\begin{proof}
In the first case $\alpha= 1$, by Theorem \ref{MGF}, there exists $\vert t_0\vert < 1/\beta$ such that $M_X(t_0)<\infty$, $X\sim\text{BWeibull}(\boldsymbol{\theta})$. 
In the second case $\alpha> 1$, again, by Theorem \ref{MGF}, there exists $t_0\in\mathbb{R}$ such that $M_X(t_0)<\infty$.
Then the proof follows.
\end{proof}

\begin{remark}
Let $X$ a continuous random variable with density function $f_X(x)$. Following the reference \cite{Klugman1998}, the rate of a random variable is given by
\begin{align*}
\tau_X=-\lim_{x\to\infty} {{\rm d} \log\big[f_X(x)\big]\over {\rm d}x}.
\end{align*}
A simple computation shows that
\begin{align*}
\tau_{\rm BWeibull(\boldsymbol{\theta})}
=
\lim_{x\to\infty}\left[\dfrac{2\delta(1-\delta x)}{1+(1-\delta x)^2}-(\alpha-1)\, {1\over x}+{\alpha\over\beta^\alpha}\, x^{\alpha-1}\right]
=
\begin{cases}
{1\over\beta} & {\rm for} \quad \alpha=1,
\\
\infty & {\rm for} \quad \alpha>1,
\\
0 & {\rm for} \quad \alpha<1.
\end{cases}
\end{align*}
Then, far enough out in the tail, every BWeibull distribution looks like
an exponential distribution when $\alpha=1$ and a Normal distribution when $\alpha>1$.
In addition, we have some comparisons between the rates of random variables with known distributions: Inverse-gamma, Log-normal, Generalized-Pareto, BWeibull, BGamma \cite{Vila2020}, exponential 
and Normal;
\begin{align*}
\tau_{{\rm InvGamma}(\alpha,\beta)}
=
\tau_{{\rm LogNorm}(\mu,\sigma^2)}
&=
\tau_{{\rm GenPareto}(\alpha,\beta, \xi)}
= 
\tau_{\rm BWeibull(\alpha<1,\beta,\delta)}=0
\\
&<
\tau_{\rm BWeibull(\alpha=1,\beta,\delta)}=
\tau_{{\rm BGamma}(\alpha,1/\beta,\delta)}=
\tau_{{\rm exp}(1/\beta)}=1/\beta
\\
&<
\tau_{\rm BWeibull(\alpha>1,\beta,\delta)}
=
\tau_{{\rm Normal}(\mu,\sigma^2)}=\infty.
\end{align*}
\end{remark}

By using the very well-known  relation $M_X(t) =
\phi_X(-it)$ between moment generating function and characteristic function of $X\sim \text{BWeibull}(\boldsymbol{\theta})$, the following result follows immediately.
\begin{proposition}
	If $X\sim \text{BWeibull}(\boldsymbol{\theta})$ then the characteristic function $\phi_X(t)=\mathbb{E}[\exp({itX})]$ can be expressed as
	\begin{align*}
	\phi_X(t)=\sum_{n=0}^{\infty} {(i\beta t)^n\over n!}\,
	\dfrac
	{\displaystyle 
		2\Gamma\bigg(1+{n\over\alpha}\bigg)+\delta^2\beta^2\Gamma\bigg(1+{n+2\over\alpha}\bigg)
		-
		2\delta\beta\Gamma\bigg(1+{n+1\over\alpha}\bigg)}
	{\displaystyle 
		2+\delta^2\beta^2\Gamma\bigg(1+{2\over\alpha}\bigg)
		-
		2\delta\beta\Gamma\bigg(1+{1\over\alpha}\bigg)},
	\quad t\in\mathbb{R}.
	\end{align*} 
\end{proposition}

\subsection{Reliability, hazard rate and the mean residual life}
\label{Reliability-hazard rate-mean residual life}
For any $t\geqslant 0$, the
reliability, the hazard rate and the mean residual life functions, associated with a random variable $X\sim \text{BWeibull}(\boldsymbol{\theta})$, are defined as follows
\begin{align*}
R(t;&\boldsymbol{\theta})
= \int_{t}^{\infty}
f(x;\boldsymbol{\theta})
\, \textrm{d} x,
\quad
H(t;\boldsymbol{\theta})
=
{f(t;\boldsymbol{\theta})\over R(t;\boldsymbol{\theta})},
\quad
\textrm{MRL}(t;\boldsymbol{\theta})
=
{1\over R(t;\boldsymbol{\theta})}
\int_{t}^{\infty}
R(x;\boldsymbol{\theta})
\, \textrm{d} x,
\end{align*}
respectively.

Let $Y\sim \text{Weibull}(\alpha,\beta)$ be a random variable with Weibull distribution. By using the integration by parts formula we have 
\begin{align*}
\mathbb{E}\big[{1}_{\{Y\geqslant t\}} Y^s\big]
=
t^s
\exp\Bigg[-\bigg({t\over\beta}\bigg)^\alpha\Bigg]
+
s
\int_{t}^{\infty}
y^{s-1} \exp\Bigg[-\bigg({y\over\beta}\bigg)^\alpha\Bigg]
\, \textrm{d} y, \quad s\geqslant 0.
\end{align*}
By taking the change of variable $x=\big({u\over\beta}\big)^\alpha$ the right-hand expression above is 
\begin{align*}
&
=
t^s
\exp\Bigg[-\bigg({t\over\beta}\bigg)^\alpha\Bigg]
+
{s\beta^s\over\alpha}
\int_{(t/\beta)^\alpha}^{\infty}
x^{{s\over\alpha}-1} \exp(-x)
\, \textrm{d} x
\\[0,2cm]
&
=
t^s
\exp\Bigg[-\bigg({t\over\beta}\bigg)^\alpha\Bigg]
+ 
{s\beta^s\over\alpha}\,
\Gamma\bigg({s\over\alpha},{t^\alpha\over\beta^\alpha}\bigg),
\end{align*}
where $\Gamma(s,x)$, $s>0$, is the upper incomplete gamma function, and where we are adopting the notation $\Gamma(0,x)=0$. Therefore,
\begin{align}\label{int-formula}
\mathbb{E}\big[{1}_{\{Y\geqslant t\}} Y^s\big]
=
t^s
\exp\Bigg[-\bigg({t\over\beta}\bigg)^\alpha\Bigg]
+ 
{s\beta^s\over\alpha}\,
\Gamma\bigg({s\over\alpha},{t^\alpha\over\beta^\alpha}\bigg),
\quad s\geqslant 0. 
\end{align}

Similarly we find that
\begin{align}\label{int-formula-1}
\mathbb{E}\big[{1}_{\{Y\leqslant t\}} Y^s\big]
=
\begin{cases}
\displaystyle
1-\exp\Bigg[-\bigg({t\over\beta}\bigg)^\alpha\Bigg] & \text{for} \quad s= 0,
\\[0,5cm]
\displaystyle
-t^s
\exp\Bigg[-\bigg({t\over\beta}\bigg)^\alpha\Bigg]
+ 
{s\beta^s\over\alpha}\,
\gamma\bigg({s\over\alpha},{t^\alpha\over\beta^\alpha}\bigg)
& \text{for} \quad s> 0, 
\end{cases}
\end{align}
where $\gamma(s,x)$, $s>0$, is the lower incomplete gamma function.

\begin{remark}
By combining \eqref{int-formula} and \eqref{int-formula-1} we have
\begin{align*}
\mathbb{E}(X^s)
&=
\mathbb{E}\big[{1}_{\{Y\leqslant t\}} Y^s\big]+\mathbb{E}\big[{1}_{\{Y\geqslant t\}} Y^s\big],
\quad s>0,
\\[0,2cm]
&=
{s\beta^s\over\alpha}\,
\bigg[
\gamma\bigg({s\over\alpha},{t^\alpha\over\beta^\alpha}\bigg)
+
\Gamma\bigg({s\over\alpha},{t^\alpha\over\beta^\alpha}\bigg)
\bigg]
=
{s\beta^s\over\alpha}\, \Gamma\bigg({s\over\alpha}\bigg)
=
{\beta^s}\, \Gamma\bigg(1+{s\over\alpha}\bigg),
\end{align*}
where in the third and fourth equalities we use the identities $\gamma(s,x)+\Gamma(s,x)=\Gamma(s)$ and $\Gamma(1+s)=s\Gamma(s)$, respectively. This confirms the validity of Proposition \ref{moments}.
\end{remark}

\begin{proposition}\label{prop-3}
	If $X\sim \text{BWeibull}(\boldsymbol{\theta})$ then
	\begin{enumerate}
		\item[1)]
		$\displaystyle
		R(t;\boldsymbol{\theta})
		=
		\dfrac{\displaystyle
			\big[1+(1-\delta t)^2\big]\,
			\exp\Bigg[-\bigg({t\over\beta}\bigg)^\alpha\Bigg]
			-
			{2\delta\beta\over\alpha}\,
			\bigg[ \Gamma\bigg({1\over\alpha},{t^\alpha\over\beta^\alpha}\bigg)
			-
			{\delta\beta}\, \Gamma\bigg({2\over\alpha},{t^\alpha\over\beta^\alpha}\bigg)
			\bigg]
		}{\displaystyle
			2+
			\delta^2\beta^2\Gamma\bigg(1+{2\over\alpha}\bigg)
			-
			2\delta\beta\Gamma\bigg(1+{1\over\alpha}\bigg)};
		$
		\item[2)] 
		$
		F(t;\boldsymbol{\theta})
		=
		\dfrac{\displaystyle
			2-\big[1+(1-\delta t)^2\big]\,
			\exp\Bigg[-\bigg({t\over\beta}\bigg)^\alpha\Bigg]
			-
			{2\delta\beta\over\alpha}\,
			\bigg[ \gamma\bigg({1\over\alpha},{t^\alpha\over\beta^\alpha}\bigg)
			-
			{\delta\beta}\, \gamma\bigg({2\over\alpha},{t^\alpha\over\beta^\alpha}\bigg)
			\bigg]
		}{\displaystyle
			2+
			\delta^2\beta^2\Gamma\bigg(1+{2\over\alpha}\bigg)
			-
			2\delta\beta\Gamma\bigg(1+{1\over\alpha}\bigg)};
		$
		\item[3)] 
		$
		H(t;\boldsymbol{\theta})
		=
		\dfrac{\displaystyle
			{\alpha\over \beta }\, 
			\big[1+(1-\delta t)^2\big]\, \bigg({t\over\beta}\bigg)^{\alpha-1}
			\exp\Bigg[-\bigg({t\over\beta}\bigg)^{\alpha}\Bigg]
		}{
			\displaystyle
			\big[1+(1-\delta t)^2\big]\,
			\exp\Bigg[-\bigg({t\over\beta}\bigg)^\alpha\Bigg]
			-
			{2\delta\beta\over\alpha}\,
			\bigg[ \Gamma\bigg({1\over\alpha},{t^\alpha\over\beta^\alpha}\bigg)
			-
			{\delta\beta}\, \Gamma\bigg({2\over\alpha},{t^\alpha\over\beta^\alpha}\bigg)
			\bigg]
		}.
		$
		\item[4)] 
		\begin{align*}
			\lim_{t\to 0} H(t;\boldsymbol{\theta})=
			\begin{cases}
				0 & {\rm for} \ \alpha>1,
				\\[0,1cm] \displaystyle 
				{1\over \beta(1-\delta\beta+\delta^2\beta^2)} & {\rm for} \ \alpha=1,
				\\[0,2cm]
				+\infty & {\rm for} \ \alpha<1,
			\end{cases}
			\quad 
			\lim_{t\to +\infty} H(t;\boldsymbol{\theta})=
			\begin{cases}
				+\infty & {\rm for} \ \alpha>1,
				\\[0,1cm] \displaystyle 
				{1\over \beta} & {\rm for} \ \alpha=1,
				\\[0,2cm]
				0 & {\rm for} \ \alpha<1.
			\end{cases}
		\end{align*}
	\end{enumerate}
\end{proposition}
The hazard function $H$ in Figure \ref{fighazards}  is drawn for different values of parameters.
\begin{figure}[htbp]
	\centering
	\subfigure[Hazard of BWeibull for $\delta$]{\label{fig:haz1}\includegraphics[width=0.4\textwidth]{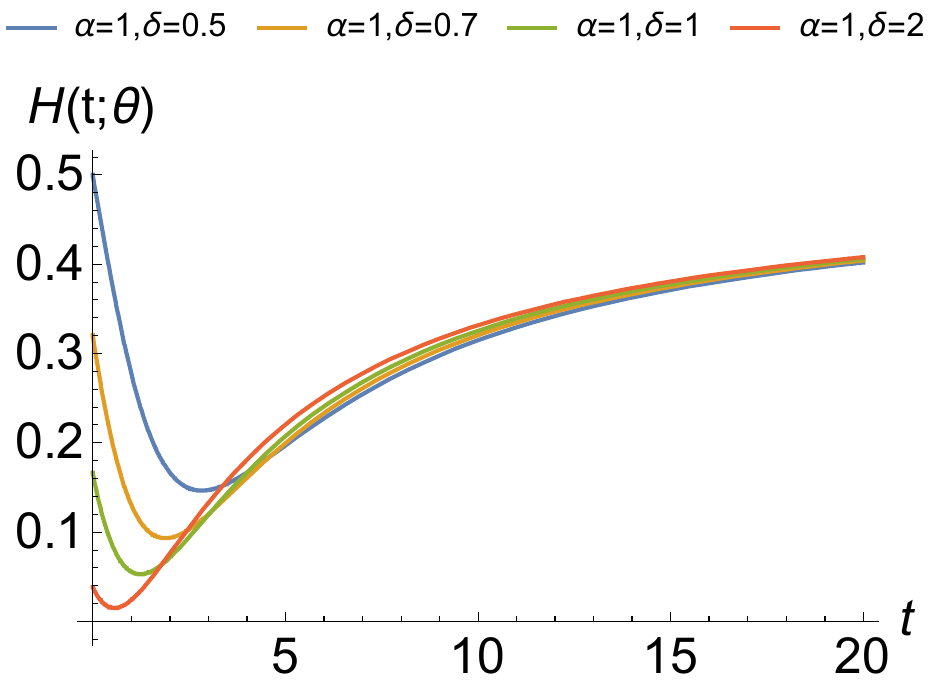}}
	\subfigure[Hazard of BWeibull for  $\alpha$ and $\delta$]{\label{fig:haz2}\includegraphics[width=0.4\textwidth]{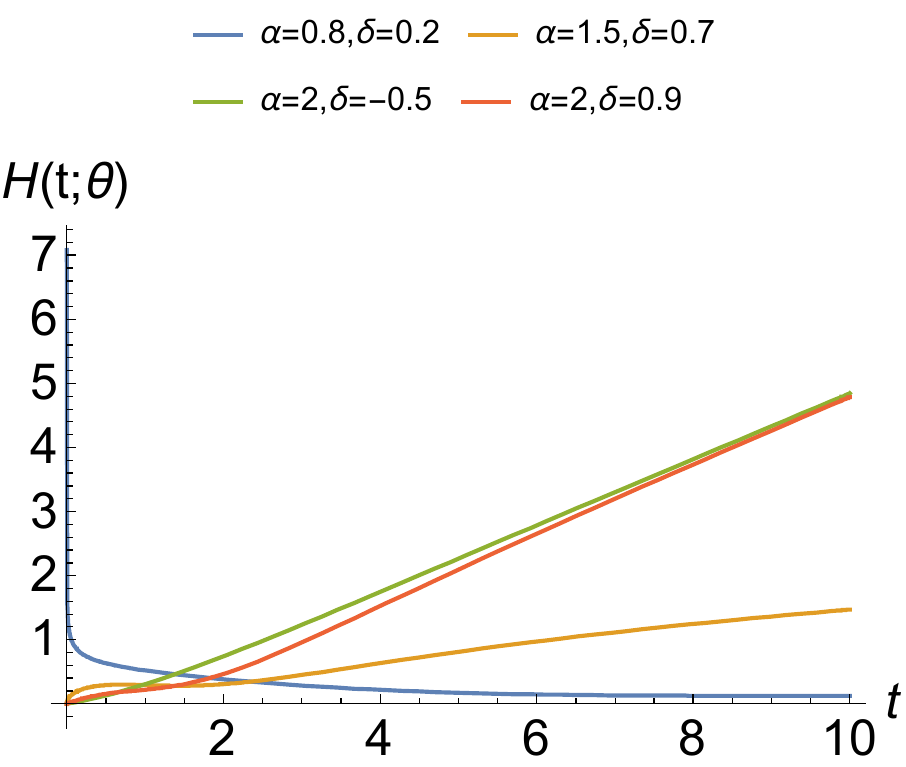}}
	\caption{Hazard of BWeibull with different values of parameters $\alpha,\delta$ and $\beta=2$.}
	\label{fighazards}
\end{figure}
\begin{proof}
	A simple algebraic manipulation in the definition of BWeibull density in equation \eqref{def-Weibull} shows that
	\begin{align*}
		\begin{array}{lllll}
			& \displaystyle
			R(t;\boldsymbol{\theta})
			=
			\dfrac{1}{Z_{\boldsymbol{\theta}}}\,\left\{
			2 \, \mathbb{E}\big[{1}_{\{Y\geqslant t\}}\big]
			-
			2\delta \, \mathbb{E}\big[{1}_{\{Y\geqslant t\}} Y\big]
			+
			\delta^2 \, \mathbb{E}\big[{1}_{\{Y\geqslant t\}} Y^2\big]\right\},
			\\[0,4cm]
			& \displaystyle
			F(t;\boldsymbol{\theta})
			=
			\dfrac{1}{Z_{\boldsymbol{\theta}}}\,\left\{
			2 \, \mathbb{E}\big[{1}_{\{Y\leqslant t\}}\big]
			-
			2\delta \, \mathbb{E}\big[{1}_{\{Y\leqslant t\}} Y\big]
			+
			\delta^2 \, \mathbb{E}\big[{1}_{\{Y\leqslant t\}} Y^2\big]\right\},
		\end{array}
		\quad 
		Y\sim \text{Weibull}(\alpha,\beta),
	\end{align*}
	where $Z_{\boldsymbol{\theta}}$ is as in \eqref{partition}.
	By taking $s=0,1,2$ in \eqref{int-formula}, we get the formula of Item 1).
	By taking $s=0,1,2$ in \eqref{int-formula-1}, the formula of Item 2) follows. 
	The proof of Item 3) follows immediately by combining the definition
	of $H(t;\boldsymbol{\theta})$ with Item 1).
	The proof of Item 4) follows directly by analyzing the form of the hazard function $H(t;\boldsymbol{\theta})$ in Item 3) through the known limits 
	$\lim_{x\to+\infty}\Gamma(s,x)=0$, $\lim_{x\to 0}\Gamma(s,x)=\Gamma(s)$, and standard limit calculations.
\end{proof}

\begin{proposition}\label{prop-exp-ind}
	If $X\sim \text{BWeibull}(\boldsymbol{\theta})$ then
	{\scalefont{0.95}
\begin{align*}
\mathbb{E}\big[{1}_{\{X\geqslant t\}} X\big]
=
		\dfrac{\displaystyle
	t\big[1+(1-\delta t)^2\big]\,
	\exp\Bigg[-\bigg({t\over\beta}\bigg)^\alpha\Bigg]
	+
	{\beta\over\alpha}\,
	\bigg[ 2\Gamma\bigg({1\over\alpha},{t^\alpha\over\beta^\alpha}\bigg)
	-
	4{\delta\beta}\, \Gamma\bigg({2\over\alpha},{t^\alpha\over\beta^\alpha}\bigg)
	+
	3 {\delta^2\beta^2}\, \Gamma\bigg({3\over\alpha},{t^\alpha\over\beta^\alpha}\bigg)
	\bigg]
}{\displaystyle
	2+
	\delta^2\beta^2\Gamma\bigg(1+{2\over\alpha}\bigg)
	-
	2\delta\beta\Gamma\bigg(1+{1\over\alpha}\bigg)}.
\end{align*}
}
\end{proposition}
\begin{proof}
By a similar decomposition to that of Proposition \ref{prop-3}, we have
	\begin{align*}
	\mathbb{E}\big[{1}_{\{X\geqslant t\}} X\big]
	=
	\dfrac{1}{Z_{\boldsymbol{\theta}}}\,\left\{
	2\, \mathbb{E}\big[{1}_{\{Y\geqslant t\}} Y\big]
	-
	2\delta\, \mathbb{E}\big[{1}_{\{Y\geqslant t\}} Y^2\big]
	+
	\delta^2\, \mathbb{E}\big[{1}_{\{Y\geqslant t\}} Y^3\big]
	\right\},
	\quad
	Y\sim \text{Weibull}(\alpha,\beta),
	\end{align*}
	where $Z_{\boldsymbol{\theta}}$ is as in \eqref{partition}.
Hence, by taking $s=1,2,3$ in \eqref{int-formula}, the proof follows.
\end{proof}

\begin{proposition}[Mean residual life function]\label{rem-1}
	If $X\sim \text{BWeibull}(\boldsymbol{\theta})$ then
\begin{align*}
	\textrm{MRL}(t;\boldsymbol{\theta})
	=
	\dfrac{\displaystyle
	{2\beta(1+\delta t)}\,
	\Gamma\bigg({1\over\alpha},{t^\alpha\over\beta^\alpha}\bigg)
	-
	{2\delta\beta^2(2+\delta t)}\, \Gamma\bigg({2\over\alpha},{t^\alpha\over\beta^\alpha}\bigg)
	+
	3 {\delta^2\beta^2}\, \Gamma\bigg({3\over\alpha},{t^\alpha\over\beta^\alpha}\bigg)
}{\displaystyle
		\alpha\big[1+(1-\delta t)^2\big]\,
\exp\Bigg[-\bigg({t\over\beta}\bigg)^\alpha\Bigg]
-
{2\delta\beta}\,
\bigg[ \Gamma\bigg({1\over\alpha},{t^\alpha\over\beta^\alpha}\bigg)
-
{\delta\beta}\, \Gamma\bigg({2\over\alpha},{t^\alpha\over\beta^\alpha}\bigg)
\bigg]}.
\end{align*}
\end{proposition}
\begin{proof}
	Integration by parts gives
\[
\mathbb{E}\big[{1}_{\{X\geqslant t\}} X\big]
=
t R(t;\boldsymbol{\theta})
+
\int_t^{\infty}
R(x;\boldsymbol{\theta})
\, {\rm d} x,
\]
because $x R(x;\boldsymbol{\theta})\to 0$ as $x\to\infty$. Then
\[
\textrm{MRL}(t;\boldsymbol{\theta})
=
{1\over R(t;\boldsymbol{\theta})}\,
\mathbb{E}\big[{1}_{\{X\geqslant t\}} X \big]
- t,
\]
where $R(t;\boldsymbol{\theta})$ and
$\mathbb{E}\big[{1}_{\{X\geqslant t\}} X\big]$
are given in Propositions \ref{prop-3} and \ref{prop-exp-ind}, respectively.
\end{proof}

\subsection{Entropies for continuous measure}
The Tsallis \cite{Tsallis1988}, Quadratic \cite{Rao2010} and  Shannon \cite{Shannon1948} entropies associated with a non-negative random variable $X$ are defined by
\begin{align}\label{Tsallis-entropy}
&S_q(X)=\dfrac{1}{q-1}\, \bigg[1-\int_{0}^{\infty} f^q(x;\boldsymbol{\theta}) \, {\rm d}x\bigg], \quad q\in\mathbb{R},
\\[0,3cm]
&H_2(X) \label{quadratic-entropy}
=
-\log\bigg[\int_{0}^{\infty} f^2(x;\boldsymbol{\theta}) \, {\rm d}x\bigg],
\\[0,3cm]
&
H_1(X)
=
-\int_{0}^{\infty} f(x;\boldsymbol{\theta}) \log\big[f(x;\boldsymbol{\theta})\big] \, {\rm d}x,
\label{Shannon-entropy}
\end{align}
respectively.

\begin{theorem}[Tsallis entropy]
Let $X\sim \text{BWeibull}(\boldsymbol{\theta})$ and 
\begin{align*}
\sum_{k=0}^{\infty} 
\bigg\vert\binom{q}{k} \bigg\vert
\int_{0}^{\infty} 
(1-\delta x)^{2k} \, \bigg({x\over\beta}\bigg)^{q(\alpha-1)}
\exp\Bigg[-q\bigg({x\over\beta}\bigg)^{\alpha}\Bigg]
\, {\rm d}x<\infty,
\end{align*}
where $\binom{q}{k}$ is the generalized binomial coefficient and $q\neq 1$.
If $\delta<0$ and $q(\alpha-1)>-1$, then the Tsallis entropy is given by
\begin{align*}
S_q(X)=	 
\dfrac{1}{q-1}-
\dfrac{
	\displaystyle 
	\alpha^{q-1}    \sum_{k=0}^{\infty} \sum_{l=0}^{2k} \binom{q}{k} \binom{2k}{l}  
\dfrac{(-\delta)^l\beta^{l-1}}{q^{q+(l-q+1)/\alpha}}\,
\Gamma\bigg(q+{l-q+1\over\alpha}\bigg)
}
{
\displaystyle
(q-1) 
	\beta^q\bigg[2+
	\delta^2\beta^2\Gamma\bigg(1+{2\over\alpha}\bigg)
	-
	2\delta\beta\Gamma\bigg(1+{1\over\alpha}\bigg)\bigg]^q
}.
	\end{align*}	
\end{theorem}	
\begin{proof}
By definition of Tsallis entropy \eqref{Tsallis-entropy} note that it is enough to prove that
\begin{align}\label{cond-suf}
\int_{0}^{\infty} f^q(x;\boldsymbol{\theta}) \, {\rm d}x
=
\dfrac{\alpha^q}{(\beta Z_{\boldsymbol{\theta}})^q}\,
\sum_{k=0}^{\infty} \sum_{l=0}^{2k} \binom{q}{k} \binom{2k}{l}  
\dfrac{(-\delta)^l\beta^{l-1}}{\alpha q^{q+(l-q+1)/\alpha}}\,
\Gamma\bigg(q+{l-q+1\over\alpha}\bigg),
\end{align}
whenever $q(\alpha-1)>-1$, where $Z_{\boldsymbol{\theta}}$ is the normalization constant in equation \eqref{partition}.

Indeed,  since
\begin{align*}
\int_{0}^{\infty} f^q(x;\boldsymbol{\theta}) \, {\rm d}x
=
\int_{0}^{\infty} 
{\alpha^q\over (\beta Z_{\boldsymbol{\theta}})^q}\, 
\big[1+(1-\delta x)^2\big]^q\, \bigg({x\over\beta}\bigg)^{q(\alpha-1)}
\exp\Bigg[-q\bigg({x\over\beta}\bigg)^{\alpha}\Bigg]
\, {\rm d}x,
\end{align*}
by using the Newton's generalized binomial theorem the above integral is 
\begin{align*}
=
{\alpha^q\over (\beta Z_{\boldsymbol{\theta}})^q}\, 
\int_{0}^{\infty} 
\sum_{k=0}^{\infty} 
\binom{q}{k}
(1-\delta x)^{2k} \, \bigg({x\over\beta}\bigg)^{q(\alpha-1)}
\exp\Bigg[-q\bigg({x\over\beta}\bigg)^{\alpha}\Bigg]
\, {\rm d}x \nonumber
\end{align*}
where the condition $\delta<0$ guarantees that $x\neq \delta/2$, and therefore the convergence of the above series.
Seeing that $\sum_{k=0}^{\infty} 
\big\vert\binom{q}{k} \big\vert
\int_{0}^{\infty} 
(1-\delta x)^{2k} \, \big({x\over\beta}\big)^{q(\alpha-1)}
\exp\big[-q\big({x\over\beta}\big)^{\alpha}\big]
\, {\rm d}x<\infty$ (hypothesis),  by Fubini's theorem \cite{knapp} the above integral is
\begin{align} \label{inte-pre}
&=
{\alpha^q\over (\beta Z_{\boldsymbol{\theta}})^q}\, 
\sum_{k=0}^{\infty} 
\binom{q}{k}
\int_{0}^{\infty} 
(1-\delta x)^{2k} \, \bigg({x\over\beta}\bigg)^{q(\alpha-1)}
\exp\Bigg[-q\bigg({x\over\beta}\bigg)^{\alpha}\Bigg]
\, {\rm d}x \nonumber
\\[0,2cm]
&={\alpha^q\over (\beta Z_{\boldsymbol{\theta}})^q}\, 
\sum_{k=0}^{\infty}
\sum_{l=0}^{2k}  
\binom{q}{k}
\binom{2k}{l}
(-\delta)^l 
\int_{0}^{\infty} x^l\, \bigg({x\over\beta}\bigg)^{q(\alpha-1)}
\exp\Bigg[-q\bigg({x\over\beta}\bigg)^{\alpha}\Bigg]
\, {\rm d}x,
\end{align}
where in the second equality we have used the classic binomial expansion.
By using the formula of Item (6) in  \cite{Cankaya2018}: 
\begin{align*}
\Gamma\bigg(s+{1\over\alpha}\bigg)
=
\alpha p^{\alpha s+1} 
\int_{0}^{\infty} 
y^{\alpha s}\, \exp\big[-(yp)^\alpha\big]
\, {\rm d}y,
\end{align*}
the expression in \eqref{inte-pre} is rewritten as
\begin{align*}
= \dfrac{\alpha^q}{(\beta Z_{\boldsymbol{\theta}})^q}\,
\sum_{k=0}^{\infty} \sum_{l=0}^{2k} \binom{q}{k} \binom{2k}{l}  
\dfrac{(-\delta)^l\beta^{l-1}}{\alpha q^{q+(l-q+1)/\alpha}}\,
\Gamma\bigg(q+{l-q+1\over\alpha}\bigg),
\end{align*}
whenever $q(\alpha-1)>-1$.
Therefore, combining the above identities we have proved that the statement in \eqref{cond-suf} is valid. This completes the proof of theorem. 
\end{proof}

\begin{theorem}[Quadratic entropy]
Let $X\sim \text{BWeibull}(\boldsymbol{\theta})$. 
If $\alpha>1/2$ then the quadratic entropy can be written as
{
\begin{align*}
H_2(X)
&=
2\log(\beta)
-\log(\alpha)
+2\log{\displaystyle 
	\left[2+
	\delta^2\beta^2\Gamma\bigg(1+{2\over\alpha}\bigg)
	-
	2\delta\beta\Gamma\bigg(1+{1\over\alpha}\bigg)\right]}
\\[0,2cm]
&
-\log\left[
{\displaystyle 
	{2^{1/\alpha}\over \beta}\,
	\Gamma\bigg(2-{1\over\alpha}\bigg)
	+
	{2\delta^2 \beta\over 2^{1/\alpha}}\,
	\Gamma\bigg(2+{1\over\alpha}\bigg)
	-
	{\delta^3 \beta^2\over 2^{2/\alpha}}\,
	\Gamma\bigg(2+{2\over\alpha}\bigg)
	+
	{\delta^4 \beta^3\over 2^{2+3/\alpha}}\,
	\Gamma\bigg(2+{3\over\alpha}\bigg)
	-
	2\delta
}
\right].
\end{align*}
}	
%
\end{theorem}
\begin{proof}
A simple algebraic manipulation shows that
{
\begin{align*}
\int_{0}^{\infty} f^2(x;\boldsymbol{\theta}) \, {\rm d}x
&=
\int_{0}^{\infty} 
{\alpha^2\over \beta^2 Z_{\boldsymbol{\theta}}^2}\, 
\big[1+(1-\delta x)^2\big]^2\, \bigg({x\over\beta}\bigg)^{2(\alpha-1)}
\exp\Bigg[-2\bigg({x\over\beta}\bigg)^{\alpha}\Bigg]
\, {\rm d}x
\\[0,2cm]
&=
{\alpha^2\over \beta^2 Z_{\boldsymbol{\theta}}^2}\,
\int_{0}^{\infty} 
\big(
4
-
8\delta x
+
8\delta^2 x^2
-
4\delta^3 x^3
+
\delta^4 x^4
\big)
\bigg({x\over\beta}\bigg)^{2(\alpha-1)}
\exp\Bigg[-\bigg({x\over\beta}\, 2^{1/\alpha}\bigg)^{\alpha}\Bigg]
\, {\rm d}x,
\end{align*}
where $Z_{\boldsymbol{\theta}}$ is as in \eqref{partition}.
}
By using the formula of Item (6) in reference \cite{Cankaya2018}:
\begin{align*}
\Gamma\bigg(s+{1\over\alpha}\bigg)
=
\alpha p^{\alpha s+1} 
\int_{0}^{\infty} 
y^{\alpha s}\, \exp\big[-(yp)^\alpha\big]
\, {\rm d}y,
\end{align*}
we have 
\begin{align*}
&{\alpha^2\over \beta^2 Z_{\boldsymbol{\theta}}^2}\,
\int_{0}^{\infty} 
\big(
4
-
8\delta x
+
8\delta^2 x^2
-
4\delta^3 x^3
+
\delta^4 x^4
\big)
\bigg({x\over\beta}\bigg)^{2(\alpha-1)}
\exp\Bigg[-\bigg({x\over\beta}\, 2^{1/\alpha}\bigg)^{\alpha}\Bigg]
\, {\rm d}x
\\[0,2cm]
&=
{\alpha^2\over \beta^2 Z_{\boldsymbol{\theta}}^2}\,
\left[
{\displaystyle 
	4v_{\boldsymbol{\theta}}(0)
	-
	8\delta v_{\boldsymbol{\theta}}(1)
	+
	8\delta^2 v_{\boldsymbol{\theta}}(2)
	-
	4\delta^3 v_{\boldsymbol{\theta}}(3)
	+
	\delta^4 v_{\boldsymbol{\theta}}(4)
}
\right],
\end{align*}
where
\begin{align*}
v_{\boldsymbol{\theta}}(s)
&=
\int_{0}^{\infty} 
x^s
\bigg({x\over\beta}\bigg)^{2(\alpha-1)}
\exp\Bigg[-\bigg({x\over\beta}\, 2^{1/\alpha}\bigg)^{\alpha}\Bigg]
\, {\rm d}x
\\[0,2cm]
&=
\dfrac{\beta^{s-1}}{\alpha 2^{2+(s-1)/\alpha}}\,
\Gamma\bigg(2+{s-1\over\alpha}\bigg), \quad  s>1-2\alpha.
\end{align*}

Therefore, 
\begin{align*}
\int_{0}^{\infty} f^2(x;\boldsymbol{\theta}) \, {\rm d}x
=
{\alpha^2\over \beta^2 Z^2}\,
\left[
{\displaystyle 
	4v_{\boldsymbol{\theta}}(0)
	-
	8\delta v_{\boldsymbol{\theta}}(1)
	+
	8\delta^2 v_{\boldsymbol{\theta}}(2)
	-
	4\delta^3 v_{\boldsymbol{\theta}}(3)
	+
	\delta^4 v_{\boldsymbol{\theta}}(4)
}
\right].
\end{align*}
Finally, taking logarithms to both sides of the above equation and then multiplying by $-1$, by definition \eqref{quadratic-entropy} of quadratic entropy and by standard algebraic manipulations,  we complete the proof.

\end{proof}

Let $Y\sim \text{Weibull}(\alpha,\beta)$. By taking the change of variable $x=\big({y\over\beta}\big)^\alpha$, for each $s>-\alpha$, we have 
\begin{align*}
\mathbb{E}[Y^s\log(Y)]
&=
{\alpha\over\beta}\, 
\int_{0}^{\infty} y^{s}\log(y)\,  \bigg({y\over\beta}\bigg)^{\alpha-1} \exp\Bigg[-\bigg({y\over\beta}\bigg)^{\alpha}\Bigg]\, {\rm d}y
\\[0,2cm]
&=
\beta^s\log(\beta)
\int_{0}^{\infty} x^{s\over\alpha} \exp(-x)\, {\rm d}x
+
{\beta^s\over\alpha}\, 
\int_{0}^{\infty} \log(x) x^{s\over\alpha} \exp(-x)\, {\rm d}x.
\end{align*}
Since 
\begin{align*}
&\int_{0}^{\infty} x^{s\over\alpha} \exp(-x)\, {\rm d}x
=
\Gamma\bigg(1+\dfrac{s}{\alpha}\bigg),
\\[0,2cm]
&\int_{0}^{\infty} \log(x) x^{s\over\alpha} \exp(-x)\, {\rm d}x
=
\Gamma\bigg(1+\dfrac{s}{\alpha}\bigg)
\Psi^{(0)}\bigg(1+\dfrac{s}{\alpha}\bigg),
\end{align*}
the expectation $\mathbb{E}[Y^s\log(Y)]$ can be written as follows
\begin{align}\label{id-main}
\mathbb{E}[Y^s\log(Y)]
=
\beta^s \Gamma\bigg(1+\dfrac{s}{\alpha}\bigg)
\bigg[\log(\beta)  
+
{1\over\alpha}\, 
\Psi^{(0)}\bigg(1+\dfrac{s}{\alpha}\bigg)
\bigg],
\end{align}
where $\Psi^{(m)}(z)={{\rm d}^{m+1}\log[\Gamma(z)] \over {\rm d}z^{m+1}}$ is the the polygamma function of order $m$.

\begin{proposition}\label{propo-exp-log}
If $X\sim \text{BWeibull}(\boldsymbol{\theta})$ then
{\scalefont{0.97}
\begin{align*}
\mathbb{E}[\log(X)]
&=
\dfrac{\displaystyle
2\bigg[\log(\beta)-{\gamma\over\alpha}\bigg]	
}{\displaystyle
2+
\delta^2\beta^2\Gamma\bigg(1+{2\over\alpha}\bigg)
-
2\delta\beta\Gamma\bigg(1+{1\over\alpha}\bigg)}
\\[0,2cm]
&
-
\dfrac{\displaystyle
2\delta\beta\Gamma\bigg(1+{1\over\alpha}\bigg) \bigg[\log(\beta)+{1\over\alpha}\, \Psi^{(0)}\bigg(1+{1\over\alpha}\bigg)\bigg]
	-
\delta^2\beta^2\Gamma\bigg(1+{2\over\alpha}\bigg) 
\bigg[\log(\beta)+{1\over\alpha}\, \Psi^{(0)}\bigg(1+{2\over\alpha}\bigg)\bigg]
}{
\displaystyle
2+
\delta^2\beta^2\Gamma\bigg(1+{2\over\alpha}\bigg)
-
2\delta\beta\Gamma\bigg(1+{1\over\alpha}\bigg)
},
\end{align*}
}
where $\gamma=-{{\rm d} \Gamma(x)\over {\rm d}x}\big\vert_{x=1}\approx 0.57721$ is the Euler-Mascheroni constant.
\end{proposition}
\begin{proof}
	Through a simple calculation we have 
\begin{align*}
\mathbb{E}[\log(X)]
=
\dfrac{1}{Z_{\boldsymbol{\theta}}}\,\left\{
2 \, \mathbb{E}\big[\log(Y)\big]
-
2\delta \, \mathbb{E}\big[Y \log(Y)\big]
+
\delta^2 \, \mathbb{E}\big[Y^2 \log(Y)\big]\right\},
\quad Y\sim \text{Weibull}(\alpha,\beta).
\end{align*}
Then, by taking $s = 0,1,2$ in \eqref{id-main}, from the above identity, the statement of proposition follows.
\end{proof}

\begin{theorem}[Shannon entropy]
Let $X\sim \text{BWeibull}(\boldsymbol{\theta})$,  
$\mathbb{E}(X)=\mu$, $\mathrm{Var}(X)=\sigma^2$, $G(x) = 1+(1-\delta x)^2$ for $x\geqslant 0$, and $\mu_G=\mathbb{E}[G(X)]=1+(1-\delta\mu)^2+\delta^2\sigma^2$. If 
\begin{align*}
\sum_{n=1}^\infty 
\dfrac{1}{n \mu_G^n}\, \Big\vert\mathbb{E}\big[\big(G(X)-\mu_G\big)^n\big]\Big\vert<\infty,
\end{align*}
then the Shannon entropy can be written as
{\scalefont{0.8}
\begin{align*}
&H_1(X)
=
\alpha\log(\beta)-\log(\alpha)	
-\log\big[1+(1-\delta\mu)^2+\delta^2\sigma^2\big]
\\[0,2cm]
&+
\log\bigg[2+
\delta^2\beta^2\Gamma\bigg(1+{2\over\alpha}\bigg)
-
2\delta\beta\Gamma\bigg(1+{1\over\alpha}\bigg)\bigg]
+
\dfrac
{\displaystyle 
	2+\delta^2\beta^2\Gamma\bigg(2+{2\over\alpha}\bigg)
	-
	2\delta\beta\Gamma\bigg(2+{1\over\alpha}\bigg)}
{\displaystyle 
	2+\delta^2\beta^2\Gamma\bigg(1+{2\over\alpha}\bigg)
	-
	2\delta\beta\Gamma\bigg(1+{1\over\alpha}\bigg)}
\\[0,2cm]
&
-
\sum_{n=1}^\infty
\sum_{k=0}^n 
\sum_{j=0}^k 
\sum_{l=0}^{2j}
\binom{n}{k} 
\binom{k}{j}
\binom{2j}{l}  \,
\dfrac{(-1)^{l-k+1} (\beta\delta)^l}{n \big[1+(1-\delta\mu)^2+\delta^2\sigma^2\big]^k}\,
\dfrac
{\displaystyle 
	2\Gamma\bigg(1+{l\over\alpha}\bigg)+\delta^2\beta^2\Gamma\bigg(1+{l+2\over\alpha}\bigg)
	-
	2\delta\beta\Gamma\bigg(1+{l+1\over\alpha}\bigg)}
{\displaystyle 
	2+\delta^2\beta^2\Gamma\bigg(1+{2\over\alpha}\bigg)
	-
	2\delta\beta\Gamma\bigg(1+{1\over\alpha}\bigg)}
\\[0,2cm]
&
-
(\alpha-1)\, 
\dfrac{\displaystyle
	2\bigg[\log(\beta)-{\gamma\over\alpha}\bigg] - 2\delta\beta\Gamma\bigg(1+{1\over\alpha}\bigg) \bigg[\log(\beta)+{1\over\alpha}\, \Psi^{(0)}\bigg(1+{1\over\alpha}\bigg)\bigg]+
	\delta^2\beta^2\Gamma\bigg(1+{2\over\alpha}\bigg) 
	\bigg[\log(\beta)+{1\over\alpha}\, \Psi^{(0)}\bigg(1+{2\over\alpha}\bigg)\bigg]
}{
	\displaystyle
	2+
	\delta^2\beta^2\Gamma\bigg(1+{2\over\alpha}\bigg)
	-
	2\delta\beta\Gamma\bigg(1+{1\over\alpha}\bigg)
}.
\end{align*}
}
\end{theorem}
\begin{proof}
	The Shannon entropy formula for the random variable $X\sim \text{BWeibull}(\boldsymbol{\theta})$, as defined in \ref{Shannon-entropy}, is given by
	\[
	H_1(X)= -\mathbb{E}\big[\log f(X;\boldsymbol{\theta}) \big].
	\]
	Developing the logarithm $\log[f(X;\boldsymbol{\theta})]$  in the above identity, we have
	\begin{align}\label{hdef}
	H_1(X)
	=
	\log(Z_{\boldsymbol{\theta}})+\alpha\log(\beta)-\log(\alpha)
	-
	\mathbb{E}\big[
	\log G(X)
	\big]
	-
	(\alpha-1)\mathbb{E}[\log(X)]+{1\over\beta^\alpha}\, \mathbb{E}(X^\alpha),
	\end{align} 
where $\mathbb{E}[\log(X)]$ and $\mathbb{E}(X^\alpha)$ are given in Propositions \ref{propo-exp-log} and \ref{moments}, respectively.

It is worth mentioning that the expectation $\mathbb {E}[\log G(X)]$ cannot be expressed through known mathematical functions.
To obtain an expression for $\mathbb {E}[\log G(X)]$ we consider Taylor's expansion of the function $\log[G(X)] $ around the point $\mu_G=\mathbb{E}[G(X)]=1+(1-\delta\mu)^2+\delta^2\sigma^2$. Hence
	\begin{align}\label{expansion}
	\log[G(X)]
	=
	\log(\mu_G)+\sum_{n=1}^\infty 
	\dfrac{(-1)^{n+1}}{n \mu_G^n}\, \big(G(X)-\mu_G\big)^n.
	\end{align}
	Taking expectation on both sides of \eqref{expansion} we get 
\begin{align}\label{obj-2}
\mathbb{E}[\log G(X)]
=
\log(\mu_G)+\mathbb{E}\left[\sum_{n=1}^\infty 
\dfrac{(-1)^{n+1}}{n \mu_G^n}\, \big(G(X)-\mu_G\big)^n\right].
\end{align}
Since $0\leqslant \mathbb{E}[\log G(X)]\leqslant \mathbb{E}\big[(1-\delta X)^2\big]<\infty$ and 
$\sum_{n=1}^\infty 
\frac{1}{n \mu_G^n}\, \big\vert\mathbb{E}\big[\big(G(X)-\mu_G\big)^n\big]\big\vert<\infty$ (hypothesis), by applying Fubini's Theorem \cite{knapp}, 
the expression on the right-hand side of \eqref{obj-2} is 
\begin{align*}
=
\log(\mu_G)+
\sum_{n=1}^\infty 
\dfrac{(-1)^{n+1}}{n \mu_G^n}\, \mathbb{E}\big[\big(G(X)-\mu_G\big)^n\big].
\end{align*}
By using binomial expansion repeatedly in the above expression, this is
\begin{align*}
=
\log(\mu_G)+
\sum_{n=1}^\infty
\sum_{k=0}^n 
\sum_{j=0}^k 
\sum_{l=0}^{2j}
\binom{n}{k} 
\binom{k}{j}
\binom{2j}{l}  \,
\dfrac{(-1)^{l-k+1} \delta^l}{n \mu_G^k}\, \mathbb{E}(X^l).
\end{align*}
Therefore, we get
\begin{align}\label{obj-3}
\mathbb{E}[\log G(X)]
=
\log(\mu_G)
+
\sum_{n=1}^\infty
\sum_{k=0}^n 
\sum_{j=0}^k 
\sum_{l=0}^{2j}
\binom{n}{k} 
\binom{k}{j}
\binom{2j}{l}  \,
\dfrac{(-1)^{l-k+1} \delta^l}{n \mu_G^k}\, \mathbb{E}(X^l).
\end{align}

Finally, by combining \eqref{obj-3}, Proposition \ref{moments} and \eqref{hdef}, the proof of theorem follows.
\end{proof}


\section{Inference: Estimation Methods and Fisher information}\label{inference}
\subsection{Maximum likelihood estimation method}
Maximum likelihood estimation (MLE) method is the standard approach for estimations of parameters of $f(x;\boldsymbol{\theta})$, mainly
due to the desirable asymptotic properties of consistency, efficiency and asymptotic normality \cite{Fer196,Fer296}.

The log likelihood function for $\boldsymbol{\theta}$ is given by
\begin{align}\label{likelitheta}
l(\boldsymbol{\theta}; \boldsymbol{x})
&=
\sum_{i=1}^{n} 
\left\{
\log\Big(\frac{\alpha}{\beta Z_{\boldsymbol{\theta}}}\Big) 
+ 
\log\left[1+(1-\delta x_i)^2\right]
+
(\alpha-1)\log\Big(\frac{x_i}{\beta}\Big) - \Big({x_i \over\beta}\Big)^{\alpha} 
\right\}
\\[0,2cm]  \nonumber
&=
-n\log(Z_{\boldsymbol{\theta}})+\sum_{i=1}^{n} \log\left[1+(1-\delta x_i)^2\right]
+
l_Y(\alpha,\beta; \boldsymbol{x}),
\end{align}
where $Z_{\boldsymbol{\theta}}=2+
\delta^2\beta^2\Gamma\big(1+{2\over\alpha}\big)
-
2\delta\beta\Gamma\big(1+{1\over\alpha}\big)$ 
is as in \eqref{partition} and $l_Y(\alpha,\beta; \boldsymbol{x})$, $Y\sim \text{Weibull}(\alpha,\beta)$, is the  log likelihood function for the parameter vector $(\alpha,\beta)$.

A standard calculation shows that the first-order partial derivatives of $l(\boldsymbol{\theta}; \boldsymbol{x})$ are
\begin{align}\label{abdscore}
\begin{array}{llllll}
\displaystyle
{\partial l(\boldsymbol{\theta}; \boldsymbol{x})\over\partial\alpha}
&=
\displaystyle
-{n\over Z_{\boldsymbol{\theta}}}\, {\partial Z_{\boldsymbol{\theta}}\over\partial\alpha}
+{\partial l_Y(\alpha,\beta; \boldsymbol{x})\over\partial\alpha},
\\[0,5cm]
\displaystyle
{\partial l(\boldsymbol{\theta}; \boldsymbol{x})\over\partial\beta}
&= 
\displaystyle
-{n\over Z_{\boldsymbol{\theta}}}\, {\partial Z_{\boldsymbol{\theta}}\over\partial\beta}
+{\partial l_Y(\alpha,\beta; \boldsymbol{x})\over\partial\beta},
\\[0,5cm] 
\displaystyle
{\partial l(\boldsymbol{\theta}; \boldsymbol{x})\over\partial\delta}
&=
\displaystyle
-{n\over Z_{\boldsymbol{\theta}}}\, {\partial Z_{\boldsymbol{\theta}}\over\partial\delta}
-\sum_{i=1}^{n}
{2\delta(1-\delta x_i)\over 1+(1-\delta x_i)^2},
\end{array}
\end{align}
where
\begin{align*}
{\partial l_Y(\alpha,\beta; \boldsymbol{x})\over\partial\alpha}
&=
{n\over\alpha}-n\log(\beta)+\sum_{i=1}^{n}\log(x_i)-
{1\over\beta^{\alpha}} \sum_{i=1}^{n}x_i^\alpha \log(x_i),
\\[0,2cm]
{\partial l_Y(\alpha,\beta; \boldsymbol{x})\over\partial\beta}
&=
-{n\alpha\over\beta}+{\alpha\over\beta^{\alpha+1}} \sum_{i=1}^{n}x_i^\alpha,
\end{align*}
and 
\begin{align*}
{\partial Z_{\boldsymbol{\theta}}\over\partial\alpha}
&=
-{2\delta^2\beta^2\over\alpha^2}\, \Psi^{(0)}\bigg(1+{2\over\alpha}\bigg)+ {2\delta\beta\over\alpha^2}\, \Psi^{(0)}\bigg(1+{1\over\alpha}\bigg),
\\[0,2cm]
{\partial Z_{\boldsymbol{\theta}}\over\partial\beta}
&=
2\delta^2\beta\Gamma\bigg(1+{2\over\alpha}\bigg)
-
2\delta\Gamma\bigg(1+{1\over\alpha}\bigg),
\\[0,2cm]
{\partial Z_{\boldsymbol{\theta}}\over\partial\delta}
&=
2\delta\beta^2\Gamma\bigg(1+{2\over\alpha}\bigg)
-
2\beta\Gamma\bigg(1+{1\over\alpha}\bigg).
\end{align*}
Here, $\Psi^{(m)}(z)={{\rm d}^{m+1}\log[\Gamma(z)] \over {\rm d}z^{m+1}}$ is the the polygamma function of order $m$.

The second-order partial derivatives of $l(\boldsymbol{\theta}; \boldsymbol{x})$ can be written as
\begin{align*}
{\partial^2 l(\boldsymbol{\theta}; \boldsymbol{x})\over\partial\alpha^2}
&=
{n\over Z_{\boldsymbol{\theta}}^2}\, \bigg({\partial Z_{\boldsymbol{\theta}}\over\partial\alpha}\bigg)^2
-{n\over Z_{\boldsymbol{\theta}}}\, {\partial^2 Z_{\boldsymbol{\theta}}\over\partial\alpha^2}
+{\partial^2 l_Y(\alpha,\beta; \boldsymbol{x})\over\partial\alpha^2},
\\[0,2cm]
{\partial^2 l(\boldsymbol{\theta}; \boldsymbol{x})\over\partial\beta^2}
&=
{n\over Z_{\boldsymbol{\theta}}^2}\, \bigg({\partial Z_{\boldsymbol{\theta}}\over\partial\beta}\bigg)^2
-{n\over Z_{\boldsymbol{\theta}}}\, {\partial^2 Z_{\boldsymbol{\theta}}\over\partial\beta^2}
+{\partial^2 l_Y(\alpha,\beta; \boldsymbol{x})\over\partial\beta^2},
\\[0,2cm]
{\partial^2 l(\boldsymbol{\theta}; \boldsymbol{x})\over\partial\delta^2}
&=
{n\over Z_{\boldsymbol{\theta}}^2}\, \bigg({\partial Z_{\boldsymbol{\theta}}\over\partial\delta}\bigg)^2
-{n\over Z_{\boldsymbol{\theta}}}\, {\partial^2 Z_{\boldsymbol{\theta}}\over\partial\delta^2}
-2\sum_{i=1}^{n}
{(1-2\delta x_i)[1+(1-\delta x_i)^2]+2\delta^2(1-\delta x_i)^2\over [1+(1-\delta x_i)^2]^2},
\\[0,2cm]
{\partial^2 l(\boldsymbol{\theta}; \boldsymbol{x})\over\partial\alpha\partial\beta}
&=
{\partial^2 l(\boldsymbol{\theta}; \boldsymbol{x})\over\partial\beta\partial\alpha}
=
{n\over Z_{\boldsymbol{\theta}}^2}\, {\partial Z_{\boldsymbol{\theta}}\over\partial\alpha}\, {\partial Z_{\boldsymbol{\theta}}\over\partial\beta}
-{n\over Z_{\boldsymbol{\theta}}}\, {\partial^2 Z_{\boldsymbol{\theta}}\over\partial\alpha\partial\beta}
+{\partial^2 l_Y(\alpha,\beta; \boldsymbol{x})\over\partial\alpha\partial\beta},
\\[0,2cm]
{\partial^2 l(\boldsymbol{\theta}; \boldsymbol{x})\over\partial\alpha\partial\delta}
&=
{\partial^2 l(\boldsymbol{\theta}; \boldsymbol{x})\over\partial\delta\partial\alpha}
=
{n\over Z_{\boldsymbol{\theta}}^2}\, {\partial Z_{\boldsymbol{\theta}}\over\partial\alpha}\, {\partial Z_{\boldsymbol{\theta}}\over\partial\delta}
-{n\over Z_{\boldsymbol{\theta}}}\, {\partial^2 Z_{\boldsymbol{\theta}}\over\partial\alpha\partial\delta},
\\[0,2cm]
{\partial^2 l(\boldsymbol{\theta}; \boldsymbol{x})\over\partial\beta\partial\delta}
&=
{\partial^2 l(\boldsymbol{\theta}; \boldsymbol{x})\over\partial\delta\partial\beta}
=
{n\over Z_{\boldsymbol{\theta}}^2}\, {\partial Z_{\boldsymbol{\theta}}\over\partial\beta}\, {\partial Z_{\boldsymbol{\theta}}\over\partial\delta}
-{n\over Z_{\boldsymbol{\theta}}}\, {\partial^2 Z_{\boldsymbol{\theta}}\over\partial\beta\partial\delta},
\end{align*}
where
\begin{align*}
{\partial^2 l_Y(\alpha,\beta; \boldsymbol{x})\over\partial\alpha^2}
&=
-{n\over\alpha^2}+{\log(\beta)\over\beta^{\alpha}} \sum_{i=1}^{n}x_i^\alpha \log(x_i) - 
{1\over\beta^{\alpha}} \sum_{i=1}^{n}x_i^\alpha \log^2(x_i),
\\[0,2cm]
{\partial^2 l_Y(\alpha,\beta; \boldsymbol{x})\over\partial\beta^2}
&=
{n\alpha\over\beta^2}-{\alpha(\alpha+1)\over\beta^{\alpha+2}} \sum_{i=1}^{n}x_i^\alpha,
\\[0,2cm]
{\partial^2 l_Y(\alpha,\beta; \boldsymbol{x})\over\partial\alpha\partial\beta}
&=
{\partial^2 l_Y(\alpha,\beta; \boldsymbol{x})\over\partial\beta\partial\alpha}
=
-{n\over\beta}+{\alpha\over\beta^{\alpha+1}} \sum_{i=1}^{n}x_i^\alpha \log(x_i),
\end{align*}
and
\begin{eqnarray}\label{Zscore}
\begin{array}{lllll}
\displaystyle
{\partial^2 Z_{\boldsymbol{\theta}}\over\partial\alpha^2}
&= \displaystyle
{4\delta^2\beta^2\over\alpha^3}\, \Psi^{(0)}\bigg(1+{2\over\alpha}\bigg)+ {4\delta^2\beta^2\over\alpha^4}\, \Psi^{(1)}\bigg(1+{2\over\alpha}\bigg)
-
{4\delta\beta\over\alpha^3}\, \Psi^{(0)}\bigg(1+{1\over\alpha}\bigg)- {2\delta\beta\over\alpha^4}\, \Psi^{(1)}\bigg(1+{1\over\alpha}\bigg),
\\[0,5cm]
\displaystyle
{\partial^2 Z_{\boldsymbol{\theta}}\over\partial\beta^2}
&= \displaystyle
2\delta^2\Gamma\bigg(1+{2\over\alpha}\bigg),
\\[0,5cm]
\displaystyle
{\partial^2 Z_{\boldsymbol{\theta}}\over\partial\delta^2}
&= \displaystyle
2\beta^2\Gamma\bigg(1+{2\over\alpha}\bigg),
\\[0,5cm]
\displaystyle
{\partial^2 Z_{\boldsymbol{\theta}}\over\partial\alpha\partial\beta}
&= \displaystyle
{\partial^2 Z_{\boldsymbol{\theta}}\over\partial\beta\partial\alpha}
=
-{4\delta^2\beta\over\alpha^2}\, \Psi^{(0)}\bigg(1+{2\over\alpha}\bigg)+ {2\delta\over\alpha^2}\, \Psi^{(0)}\bigg(1+{1\over\alpha}\bigg),
\\[0,5cm]
{\partial^2 Z_{\boldsymbol{\theta}}\over\partial\alpha\partial\delta}
&= \displaystyle
{\partial^2 Z_{\boldsymbol{\theta}}\over\partial\delta\partial\alpha}
=
-{4\delta\beta^2\over\alpha^2}\, \Psi^{(0)}\bigg(1+{2\over\alpha}\bigg)+ {2\beta\over\alpha^2}\, \Psi^{(0)}\bigg(1+{1\over\alpha}\bigg),
\\[0,5cm]
\displaystyle
{\partial^2 Z_{\boldsymbol{\theta}}\over\partial\beta\partial\delta}
&= \displaystyle
{\partial^2 Z_{\boldsymbol{\theta}}\over\partial\delta\partial\beta}
=
4\delta\beta\Gamma\bigg(1+{2\over\alpha}\bigg)
-2\Gamma\bigg(1+{1\over\alpha}\bigg).
\end{array}
\end{eqnarray}

\subsection{Maximum $\log_q$ likelihood estimation method}
The generalization of MLE with $\log$ form  is defined as maximum $\log_q$ likelihood estimation (MLqE). $\log_q$ is $q$-deformed logarithm of natural logarithm ($\log$) in MLE \cite{Tsallisbook09}. The concavity property of $\log_q$ guarantees to replace $\log$ in MLE by $\log_q$ \cite{Lindsay94,CanKor18,Jan1,Jan2}.
%
\begin{align}\label{likeliqtheta}
l_q(\boldsymbol{\theta}; \boldsymbol{x}) = \sum_{i=1}^{n}  \log_q\big[f(x_i;\boldsymbol{\theta})\big],
\end{align}
where $\log_q(f)={f^{1-q}-1 \over 1-q}, q \in \mathbb{R} \backslash \{1\}$. If $q \rightarrow 1$, then we have $\log$ and so equation \eqref{likelitheta} is dropped. When $\log_q$ likelihood is compared with $\log$ likelihood, the deformation gives an advantage to manage the modelling capability of a parametric model.   $\log_q(f)$ as an objective or loss function in MLqE in the M-estimation  is used to manage efficiency and robustness for contamination.  Note that $\log_q(f)$ is M-function \cite{MLqEBio,CanKor18}. This estimation method is applied to BWeibull$(\boldsymbol{\theta})$.

The $l_q$ function in equation \eqref{likeliqtheta} is optimized according to parameters $\boldsymbol{\theta}$ in order to get the estimators $\widehat{\boldsymbol{\theta}}$ of $\boldsymbol{\theta}$. Alternatively, a system of estimating equations can be solved simultaneously  according to the corresponding estimating equation for the parameters ${\alpha}$, ${\beta}$ and ${\delta}$. We omit to rewrite the system $\sum_{i=1}^{n} f(x_i;\boldsymbol{\theta})^{1-q} \, {\partial \log[f(x_i;\boldsymbol{\theta})] \over \partial \boldsymbol{\theta}}= \boldsymbol{0}$. Note that ${\partial \log[f(x_i;\boldsymbol{\theta})] \over \partial \boldsymbol{\theta}}$ is given by equation \eqref{abdscore} (with $n=1$) for $\alpha,\beta$ and $\delta$.


\subsection{Fisher information based on $\log$ and $\log_q$}
Fisher information matrix (FI) is given by
\begin{equation}\label{F2case}
 \mathbb{E}\left[
 \bigg({\partial l \over \partial \boldsymbol{\theta}}\bigg)
 \bigg({\partial l \over \partial \boldsymbol{\theta}}\bigg)^{\intercal }\, \right]
 =
 \mathbb{E}\left[{\partial^2 l \over \partial \boldsymbol{\theta} \partial \boldsymbol{\theta}^{\intercal}}\right]
\end{equation}
and it is a tool to provide the variances of $\widehat{\boldsymbol{\theta}}$, i.e., Var($\widehat{\boldsymbol{\theta}}$). As it is well-known, the inverse of FI gives  Var($\widehat{\boldsymbol{\theta}}$) if the inverse of FI matrix exists \cite{Fisher25,LehmannCas98}. If the FI matrix is singular, then the generalized inverse techniques are used to get the inverse of FI even though we have information loss due to the used generalized inverse \cite{LiYeh12gen}. 

FI matrix based on $\log_q$ is given by the following form:
\begin{equation}\label{abdF}
F=n\begin{bmatrix} 
E_{\alpha \alpha} & E_{\alpha \beta} & E_{\alpha \delta} \\
E_{\beta \alpha} & E_{\beta \beta} & E_{\beta \delta}  \\
E_{\delta \alpha} & E_{\delta \beta} & E_{\delta \delta}  
\end{bmatrix},
\end{equation}
\noindent where $n$ is sample size. $E$ is integral for partial derivatives of $\log(L)$ according to parameters and it is taken over probability density function $f(x;\boldsymbol{\theta})$. The subscript in $E$ represents second-order partial derivatives of $\log(L)$ according to parameters $\alpha, \beta$ and $\delta$. In other words, if $X\sim \text{BWeibull}(\boldsymbol{\theta})$ then
\begin{align*}
E_{\alpha \alpha}
&=
{1\over Z_{\boldsymbol{\theta}}^2}\, \bigg({\partial Z_{\boldsymbol{\theta}}\over\partial\alpha}\bigg)^2
-{1\over Z_{\boldsymbol{\theta}}}\, {\partial^2 Z_{\boldsymbol{\theta}}\over\partial\alpha^2}
+-{1\over\alpha^2}+{\log(\beta)\over\beta^{\alpha}}\, \mathbb{E}\big[X^\alpha \log(X)\big] - 
{1\over\beta^{\alpha}}\, \mathbb{E}\big[X^\alpha \log^2(X)\big],
\\[0,2cm]
E_{\beta \beta}
&= 
{1\over Z_{\boldsymbol{\theta}}^2}\, \bigg({\partial Z_{\boldsymbol{\theta}}\over\partial\beta}\bigg)^2
-{1\over Z_{\boldsymbol{\theta}}}\, {\partial^2 Z_{\boldsymbol{\theta}}\over\partial\beta^2}
+{\alpha\over\beta^2}-{\alpha(\alpha+1)\over\beta^{\alpha+2}}\, \mathbb{E}\big(X^\alpha\big),
\\[0,2cm]
E_{\delta \delta}
&= 
{1\over Z_{\boldsymbol{\theta}}^2}\, \bigg({\partial Z_{\boldsymbol{\theta}}\over\partial\delta}\bigg)^2
-{1\over Z_{\boldsymbol{\theta}}}\, {\partial^2 Z_{\boldsymbol{\theta}}\over\partial\delta^2}
-2\mathbb{E}\left[
{(1-2\delta X)\big(1+(1-\delta X)^2\big)+2\delta^2(1-\delta X)^2\over \big(1+(1-\delta X)^2\big)^2}
\right],
\\[0,2cm]
E_{\alpha \beta}
&= 
E_{\beta \alpha}
=
{1\over Z_{\boldsymbol{\theta}}^2}\, {\partial Z_{\boldsymbol{\theta}}\over\partial\alpha}\, {\partial Z_{\boldsymbol{\theta}}\over\partial\beta}
-{1\over Z_{\boldsymbol{\theta}}}\, {\partial^2 Z_{\boldsymbol{\theta}}\over\partial\alpha\partial\beta}
-{1\over\beta}+{\alpha\over\beta^{\alpha+1}}\, \mathbb{E}\big[X^\alpha \log(X)\big],
\\[0,2cm]
E_{\alpha \delta}
&= 
E_{\delta \alpha}
=
{1\over Z_{\boldsymbol{\theta}}^2}\, {\partial Z_{\boldsymbol{\theta}}\over\partial\alpha}\, {\partial Z_{\boldsymbol{\theta}}\over\partial\delta}
-{1\over Z_{\boldsymbol{\theta}}}\, {\partial^2 Z_{\boldsymbol{\theta}}\over\partial\alpha\partial\delta},
\\[0,2cm]
E_{\beta \delta}
&= 
E_{\delta \beta}
=
{1\over Z_{\boldsymbol{\theta}}^2}\, {\partial Z_{\boldsymbol{\theta}}\over\partial\beta}\, {\partial Z_{\boldsymbol{\theta}}\over\partial\delta}
-{1\over Z_{\boldsymbol{\theta}}}\, {\partial^2 Z_{\boldsymbol{\theta}}\over\partial\beta\partial\delta},
\end{align*}
where ${\partial^2 Z_{\boldsymbol{\theta}}\over\partial\theta\partial\theta'}$, $\theta,\theta' \in\{\alpha,\beta,\delta\}$, are given in Item \eqref{Zscore}  and 
\begin{align}\label{xalogx}
&\mathbb{E}\big[X^\alpha \log(X)\big]
=
\dfrac{
\beta^{\alpha} 
\big[
2-2\gamma + 2\alpha \log(\beta)
\big]
}{
\displaystyle 
{\alpha\over 2} \bigg[
2+
\delta^2\beta^2\Gamma\bigg(1+{2\over\alpha}\bigg)
-
2\delta\beta\Gamma\bigg(1+{1\over\alpha}\bigg)\bigg]-\alpha
}
\\[0,2cm] 
\nonumber
&+
\dfrac{ \displaystyle
\beta^{\alpha+1} \delta 
\bigg\{
-2\Gamma\bigg(2+{1\over\alpha}\bigg)
\bigg[
\alpha\log(\beta) 
+
\psi^{(0)}\bigg(2+{1\over\alpha}\bigg)  
\bigg]
+
\beta \delta \Gamma\bigg(2 + {2\over\alpha}\bigg)
\bigg[\alpha\log(\beta)+\psi^{(0)}\bigg(2+{1\over\alpha}\bigg) \bigg] 
\bigg\}
}{\displaystyle 
{\alpha\over 2} \bigg[
2+
\delta^2\beta^2\Gamma\bigg(1+{2\over\alpha}\bigg)
-
2\delta\beta\Gamma\bigg(1+{1\over\alpha}\bigg)\bigg]-\alpha},
\end{align}

\begin{align}\label{xalog2x}
&\mathbb{E}\big[X^\alpha \log^2(X)\big]
=
\dfrac{\displaystyle
\beta^{\alpha}\big[
-12 \gamma + 6 \gamma^2 + \pi^2 + 12 \alpha \log(\beta) - 12 \gamma \alpha \log(\beta)
+6 \alpha^2 \log(\beta)^2 
\big]
}{\displaystyle 3\alpha^2 \bigg[
2+
\delta^2\beta^2\Gamma\bigg(1+{2\over\alpha}\bigg)
-
2\delta\beta\Gamma\bigg(1+{1\over\alpha}\bigg)\bigg]}
\\[0,2cm]
\nonumber
&
-
\dfrac{\displaystyle
6\beta^{\alpha+1} \delta \Gamma\bigg(2+{1\over\alpha}\bigg) 
\bigg[
\alpha^2 \log^2(\beta)+2\alpha \log(\beta)\psi^{(0)}\bigg(2+{1\over\alpha}\bigg) + \psi^{(0)}\bigg(2+{1\over\alpha}\bigg)^2 + \psi^{(1)}\bigg(2+{1\over\alpha}\bigg)
\bigg]
}{\displaystyle 3\alpha^2 \bigg[
2+
\delta^2\beta^2\Gamma\bigg(1+{2\over\alpha}\bigg)
-
2\delta\beta\Gamma\bigg(1+{1\over\alpha}\bigg)\bigg]}
\\[0,2cm]
\nonumber
&+
\dfrac{ \displaystyle
3\beta^{\alpha+2}\delta^2 \Gamma\bigg(2+{2\over\alpha}\bigg) 
\bigg[
\alpha^2 \log(\beta)^2 + 2 \alpha \log(\beta)\psi^{(0)}\bigg(2+{2\over\alpha}\bigg)+\psi^{(0)}\bigg(2+{2\over\alpha}\bigg)^2+\psi^{(1)}\bigg(2+{2\over\alpha}\bigg)
\bigg]
}{\displaystyle  3\alpha^2 \bigg[
2+
\delta^2\beta^2\Gamma\bigg(1+{2\over\alpha}\bigg)
-
2\delta\beta\Gamma\bigg(1+{1\over\alpha}\bigg)\bigg]},
\end{align}
where
$\gamma=-{{\rm d} \Gamma(x)\over {\rm d}x}\big\vert_{x=1}\approx 0.57721$ is the Euler-Mascheroni constant, and
\begin{align}\label{xa}
\mathbb{E}\big(X^\alpha \big)
= 
\beta^{\alpha}\,
\dfrac{
2+\beta \delta
\bigg[\beta \delta \Gamma\bigg(2+\dfrac{2}{\alpha}\bigg)-2\Gamma\bigg(2+\dfrac{1}{\alpha}\bigg)\bigg]
}{\displaystyle 
2+
\delta^2\beta^2\Gamma\bigg(1+{2\over\alpha}\bigg)
-
2\delta\beta\Gamma\bigg(1+{1\over\alpha}\bigg)}.
\end{align}
The proofs of equations \eqref{xalogx}-\eqref{xa} are followed by the integral kernel given by item (6) \cite{Cankaya2018} and taking derivative of both sides of integral or the software Mathematica 12.0$\copyright$ can be used.
 
The definition of Fisher information based on $\log_q$ is given by, see \cite{Plastinoetal97,CanKor18}, 
\begin{align}\label{qFisher}
_qF(\boldsymbol{\theta})
\Big\vert_{\boldsymbol{\theta}\coloneqq\widehat{\boldsymbol{\theta}}}	
=  
_q\mathbb{E}
\left[{\partial l \over \partial \boldsymbol{\theta}}  \left( {\partial l\over \partial \boldsymbol{\theta}}  \right)^{\intercal}  f^{1-q} (X;\boldsymbol{\theta}) \right] 
= 
\int_{0}^\infty {\partial l \over \partial \boldsymbol{\theta}}   
\left({\partial \l \over \partial \boldsymbol{\theta}} \right)^{\intercal} 
f^{2-q}(x;\boldsymbol{\theta}) \,  \ud x,
\end{align}
where $l = \sum_{i=1}^{n} \log[f(x_i;\boldsymbol{\theta})]$ and $f(x;\boldsymbol{\theta})$ is a parametric model.  The calculation of integral is performed for one variable case $x_i$ from observations $x_1,x_2,\dots,x_n$. Since it is replicated for $n$ case, FI matrix can be rewritten as the form in equation \eqref{abdF}.  If $q=1$ in equation \eqref{qFisher}, then $_qF$ drops to $F$ in equation \eqref{abdF}. The connection between Fisher information and Tsallis $q$-entropy is proved by \cite{CanKor18}. Since the analytical tractability of integral calculation cannot be easy to follow, the numerical integration technique in Mathematica 12.0$\copyright$ is used to get the values of elements in matrix in  equation \eqref{qFisher}. 

\section{Application on real data sets}\label{appsection}

The real data sets are modelled by using BWeibull and BGamma distributions which have bimodality property. The parameter $\delta$ that gives an advantage to model the bimodal data  which can occur due to contamination or irregularity into underlying or main distribution as an indicator for abnormalization in the working principle of a phenomena in the universe is used in Weibull and Gamma distributions. As it is clearly observed, the parameter $\delta$ with $1+(1-\delta x)^2$ produces a modality due to fact that the function is polynomial with order $2$. We confine ourself the real data sets which can be modelled by function having the  bimodality property. 

MLE and MLqE methods are used to estimate the model parameters.  If a data set has two peaks, then it is expected that BWeibull and BGamma distributions capture the peaks and so they are flexible when they are compared with Weibull and Gamma distributions. Since BWeibull and BGamma distributions have analytical expression for their CDF, it is advantegous to use them for getting  statistics from Kolmogorov-Smirnov (KS) and Cram\'{e}r–von Mises (CVM). Since we use $\log$ and $\log_q$ in MLE and MLqE methods for getting the estimates of parameters, using information criteria (IC) such as Akaike and Bayesian, etc is not appropriate in view of the fact that $\log$ and $\log_q$ are not comparable. As it is well-known \cite{Hamparsum87}, the lack of fit part of IC depends on $\log$. There should be an equilibrium among the values of IC applied to different probability density functions. Since there does not exist have the equilibrium between $\log$ and $\log_q$ to make a comparison for lack of fit part of IC, it is not appropriate to use IC and we consult goodness of fit tests which are KS and CVM in order to test the modelling competence of the used objective functions produced by $\log(f)$ and $\log_q(f)$ \cite{God60,Hub64} as objective functions from MLE and MLqE, respectively. Note that  MLE and MLqE are methods used to estimate the parameters of $f(x;\boldsymbol{\theta})$. However, $\log(f)$ are alternative step for calculation in order to reach the estimators $\widehat{\boldsymbol{\theta}}$ from MLE. This alternative situation is a way to generalize MLE method as M-estimation method \cite{God60,Hub64}. $q$ in $\log_q$ is a parameter to derive the different forms of $f(x;\boldsymbol{\theta})$ \cite{Bercher12a}. For this reason,   $\log_q$ likelihood and its special form ($\log$ likelihood) are applied to estimate the parameters $\boldsymbol{\theta}=(\alpha,\beta,\delta)$.

The functions in  equations \eqref{likelitheta}-\eqref{likeliqtheta} are nonlinear and the estimating equations for the parameters $\alpha,\beta$ and $\delta$ are also nonlinear  to get the estimators $\widehat{\alpha},\widehat{\beta}$ and $\widehat{\delta}$ analytically.  The stochasticity and metaheuristic algorithms are tools used to optimize the nonlinear functions according to the parameters to get the estimates of $\widehat{\alpha},\widehat{\beta}$ and $\widehat{\delta}$.  The optimization of functions in equations \eqref{likelitheta}-\eqref{likeliqtheta} is performed by a 'metaheuristicOpt' package with 'metaOpt' function in free statistical software $R$ version $4.0.2$. The chosen algorithm in 'metaOpt' is Harmony Search (HS) \cite{harmony} due to fact that we have the highest p-values of KS and CVM, which can show that BWeibull and BGamma with $\log$ and $\log_q$ should be optimized by 'metaOpt' with HS. After applying the 'metaOpt' with HS to functions in equations \eqref{likelitheta} and \eqref{likeliqtheta}, we have the estimated values of $\widehat{\boldsymbol{\theta}}$ from $\log$ and $\log_q$ likelihood estimaton methods. The parameter $q$ in $\log_q$ likelihood estimation method is chosen according the biggest values of probability (p-value) for KS and CVM test statistics. Note that the objective function $\log_q(f)$ takes different forms for each value of $q$. Thus, for only one parametric model $f(x;\boldsymbol{\theta})$, we can have neighborhoods  of $f(x;\boldsymbol{\theta})$ for different value of $q$, which helps us to manage the efficiency and robustness while performing the modelling on data. 

BWeibull and BGamma distributions with their $\log$, $\log_q$ forms, i.e., $\log(f)$ and $\log_q(f)$ as objective functions  are applied at $5$ real data sets in order to test the modelling competence on the data sets.  The c.d. and p.d. functions abbreviated as CDF and PDF are superimposed on the empirical c.d. function and histogram of data to depict the best modelling chosen according to the p-values of KS and CVM (see Figures \ref{figcarbonfibers}-\ref{figgastric}). In Figures \ref{figcarbonfibers}(a)-\ref{figgastric}(a), CDF is more accurate illustration to depict the modelling of function with estimates and also the illustrations of CDF and PDF from Figures \ref{figcarbonfibers}-\ref{figgastric}) show that the modelling on data sets is accomplished well. Using the probability values of KS and CVM is necessary to provide  the homogenity for comparison with their corresponding test statistics of KS and CVM.

{\bf Example 1 }: The data are  the breaking stress of carbon fibers. First 10 subgroups in control which consists of $50$ observations are modelled by using Weibull distribution \cite{carbonfibers}.  Table \ref{tablecarbonfibers} represents the values of estimates for the carbon fibers data modelled by BWeibull and BGamma with $\log$ and $\log_q$ as objective functions. The role of MLqE is observed. That is, there is an important increasing on the modelling competence when $\log_q(f)$ from MLqE is used. Note that it is observed that there is a good accommodation between this data as an empricial distribution and $\log_q(f)$ from MLqE, which can lead to have the highest p-value of KS. The parameter $q$ as tuning constant produces the neighborhoods of a parametric model, i.e., objective function, in order to increase the modelling competence of a parametric model \cite{Bercher12a}. 

\input{table1}
\begin{figure}[htb!]
  \centering
  \subfigure[(Empirical) CDF of BWeibull($\widehat{{\theta}}_{\text{MLqE}}$)]{\label{fig:carbonfibersCDFq}\includegraphics[width=0.41\textwidth]{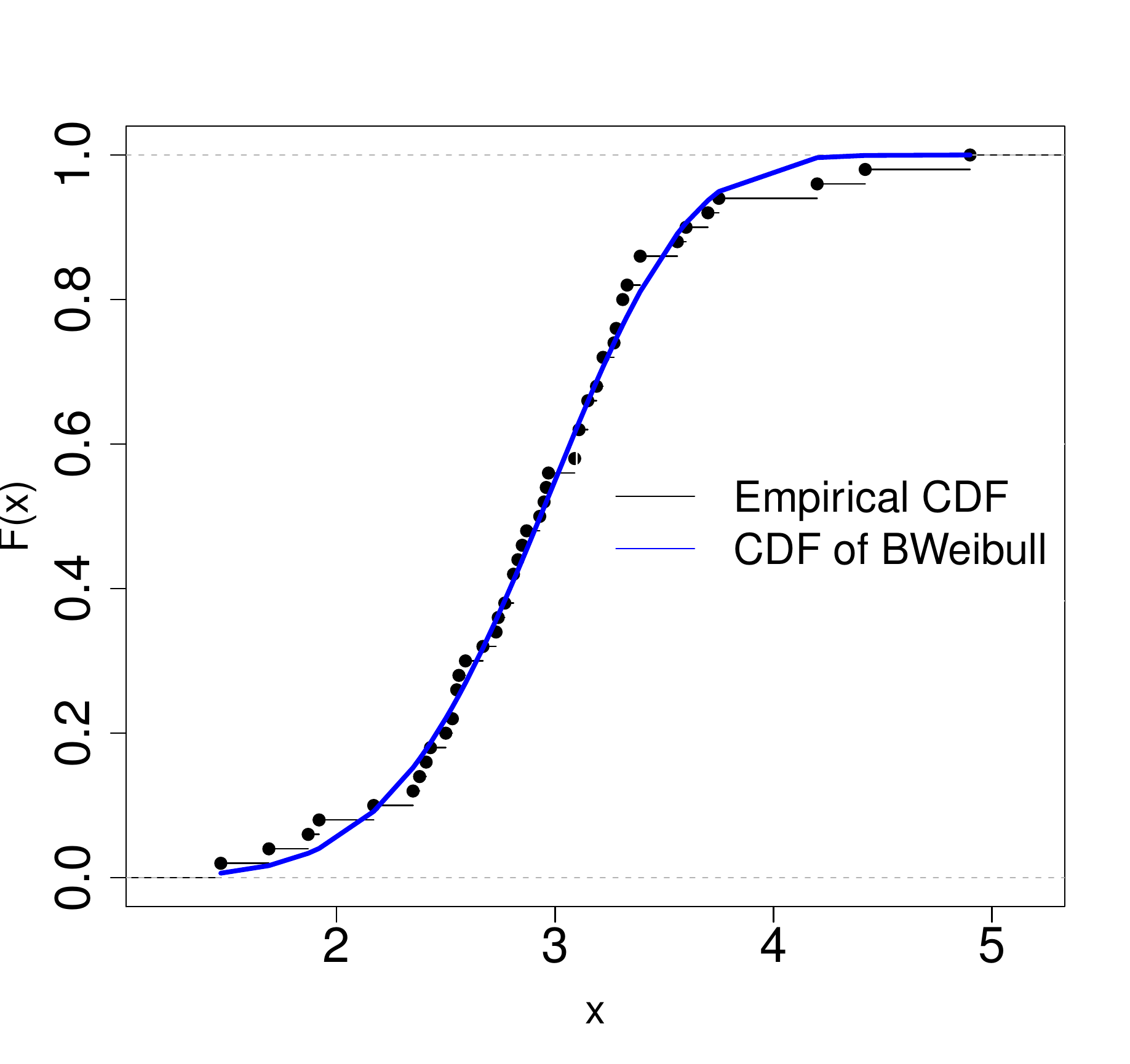}}
  \subfigure[Frequency and PDF of BWeibull($\widehat{{\theta}}_{\text{MLqE}}$)]{\label{fig:carbonfibersPDFq}\includegraphics[width=0.41\textwidth]{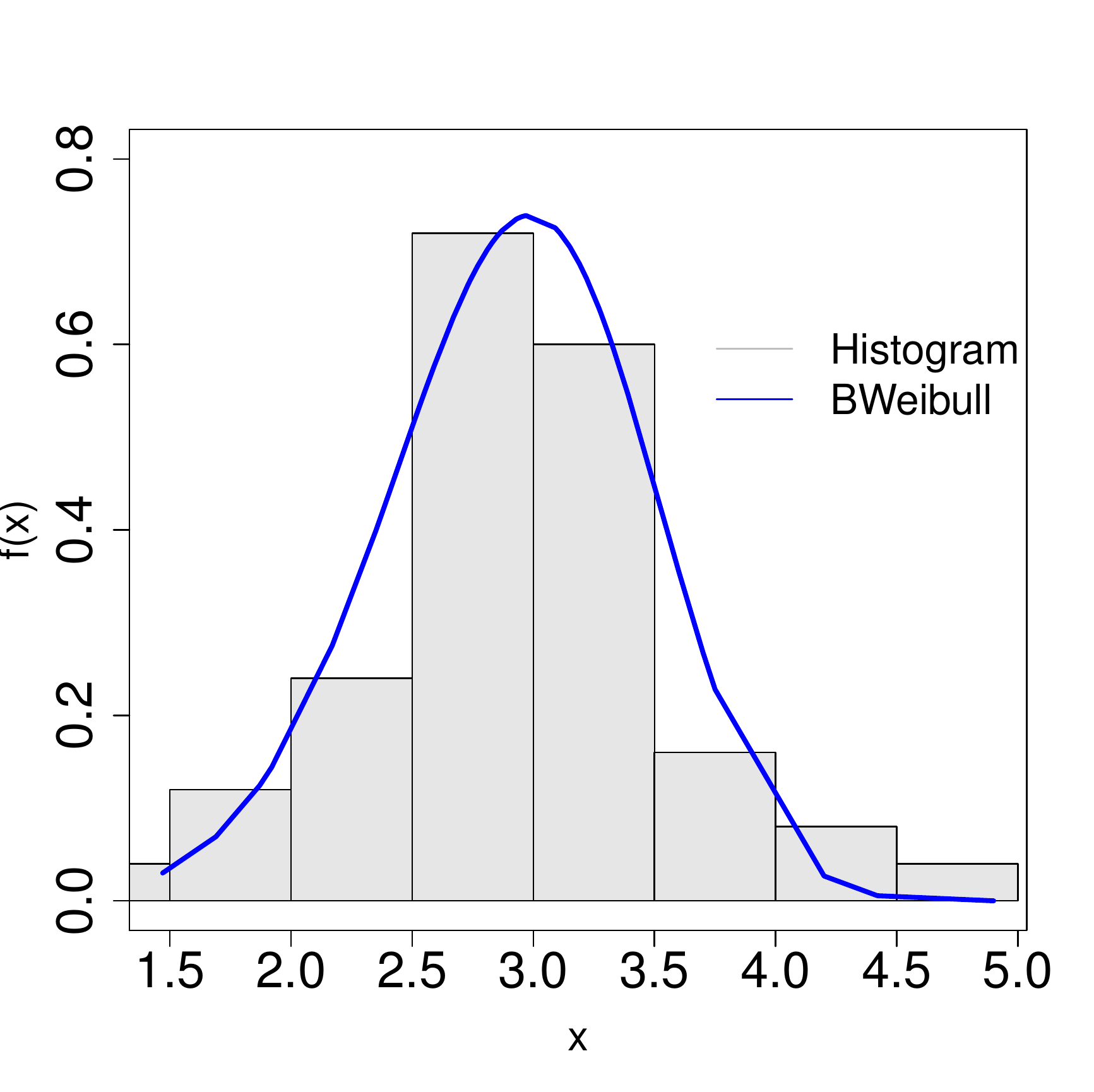}}
  \caption{The fitted values from MLqE of parameters of bimodal Weibull distribution for  carbon fibers data}
  \label{figcarbonfibers}
\end{figure}

{\bf Example 2 }: The data are maximum ozone concentrations labeled as "o3max" in "goft" package in free statistical software $R$ version $4.0.2$. 24 observations are included by o3max.  Table \ref{tableo3max} shows  MLE and MLqE with BWeibull are better than that of BGamma for the p-values of KS and CVM. Note that BWeibull with $\log$ and $\log_q$ has an important fitting competence when compared with that of BGamma.  
\input{table2}

\begin{figure}[htbp]
	\centering
	\subfigure[(Empirical) CDF of BWeibull($\widehat{{\theta}}_{\text{MLqE}}$)]{\label{fig:o3maxCDFq}\includegraphics[width=0.41\textwidth]{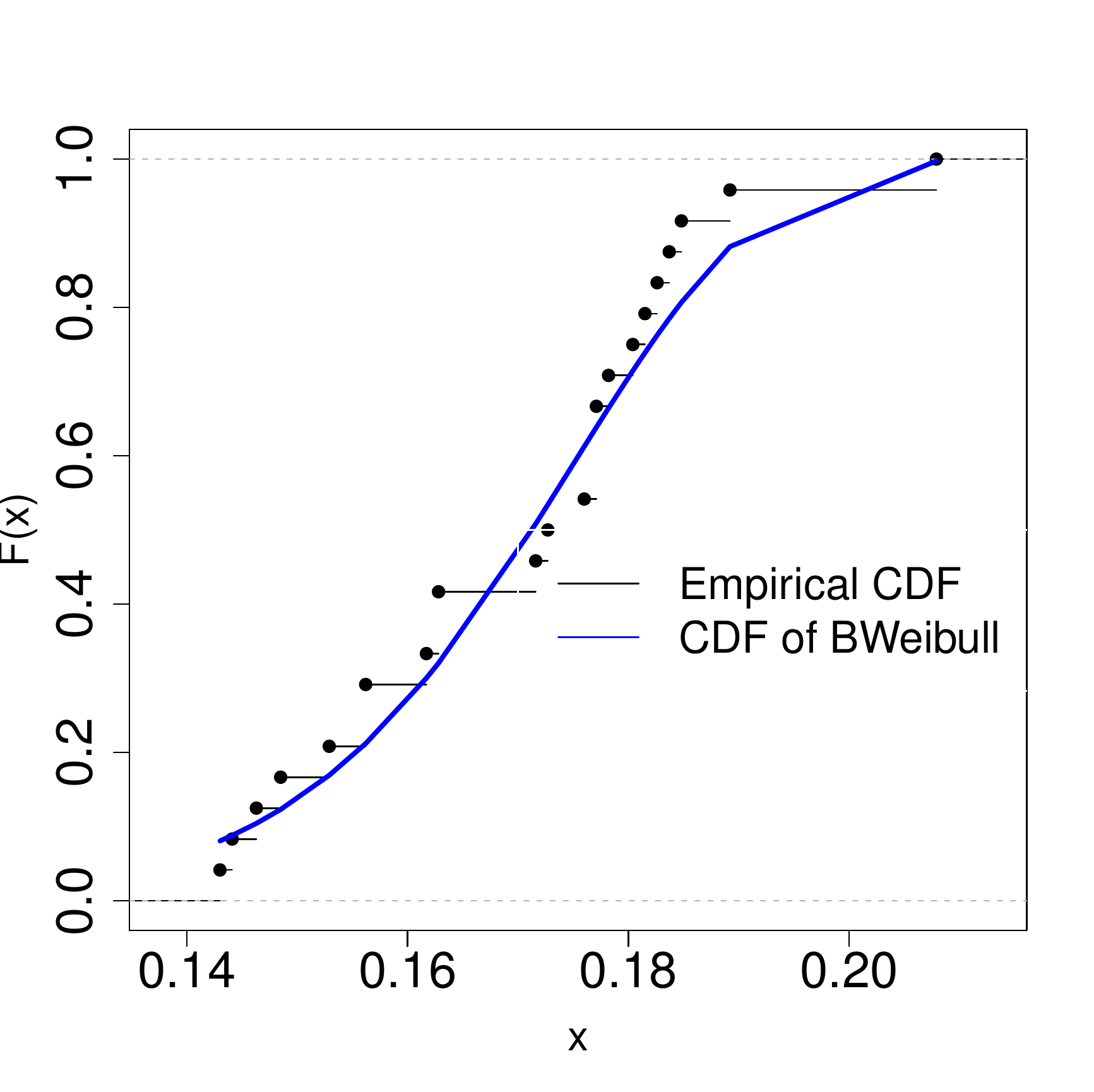}}
	\subfigure[Frequency and PDF of BWeibull($\widehat{{\theta}}_{\text{MLqE}}$)]{\label{fig:o3maxPDFq}\includegraphics[width=0.41\textwidth]{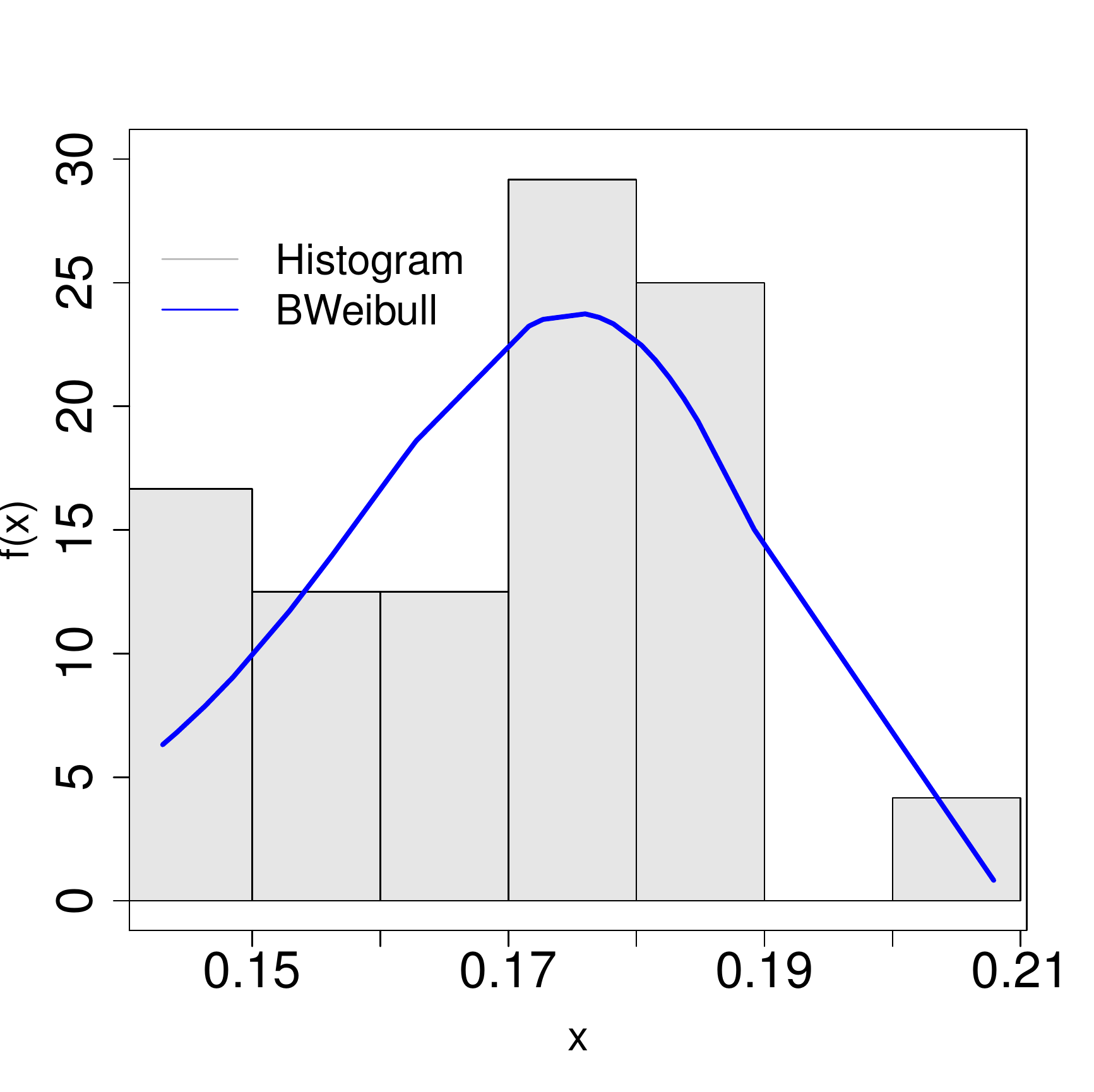}}
	\caption{The fitted values from MLqE of parameters of bimodal Weibull distribution for o3max data}
  \label{figo3max}
\end{figure}

{\bf Example 3 }: Growth hormone data was also modelled by \cite{Alizadeh}. The number of observation is 35. MLqEs with BWeibull and BGamma distributions outperform when compared with parametric model in  \cite{Alizadeh}.  
\input{table3}
Table \ref{tableghormone} represents the values of estimates for growth hormone data modelled by objective functions. {\it Italic} represents comparison between CVM values when  MLE method is used. Note that MLE method reflects directly parametric models for which BWeibull and BGamma are used. In other words, deformation on the parametric model does not exist when we compare with $\log_q$. When MLEs of parameters of BWeibull and BGamma are compared, we observe that MLE of parameters of BWeibull is a little ahead for the performance of the fitting competence according to the values of CVM and its corresponding p-values. KS statistic and the corresponding p-values are not capable to assess the modelling competence. According to the values of p-value(CVM), MLqE with BGamma  give the best performance on the modelling data when compared with MLqE with BWeibull. 

\begin{figure}[htbp]
	\centering
	\subfigure[(Empirical) CDF of BWeibull($\widehat{{\theta}}_{\text{MLE}}$)]{\label{fig:hormoneCDFq}\includegraphics[width=0.41\textwidth]{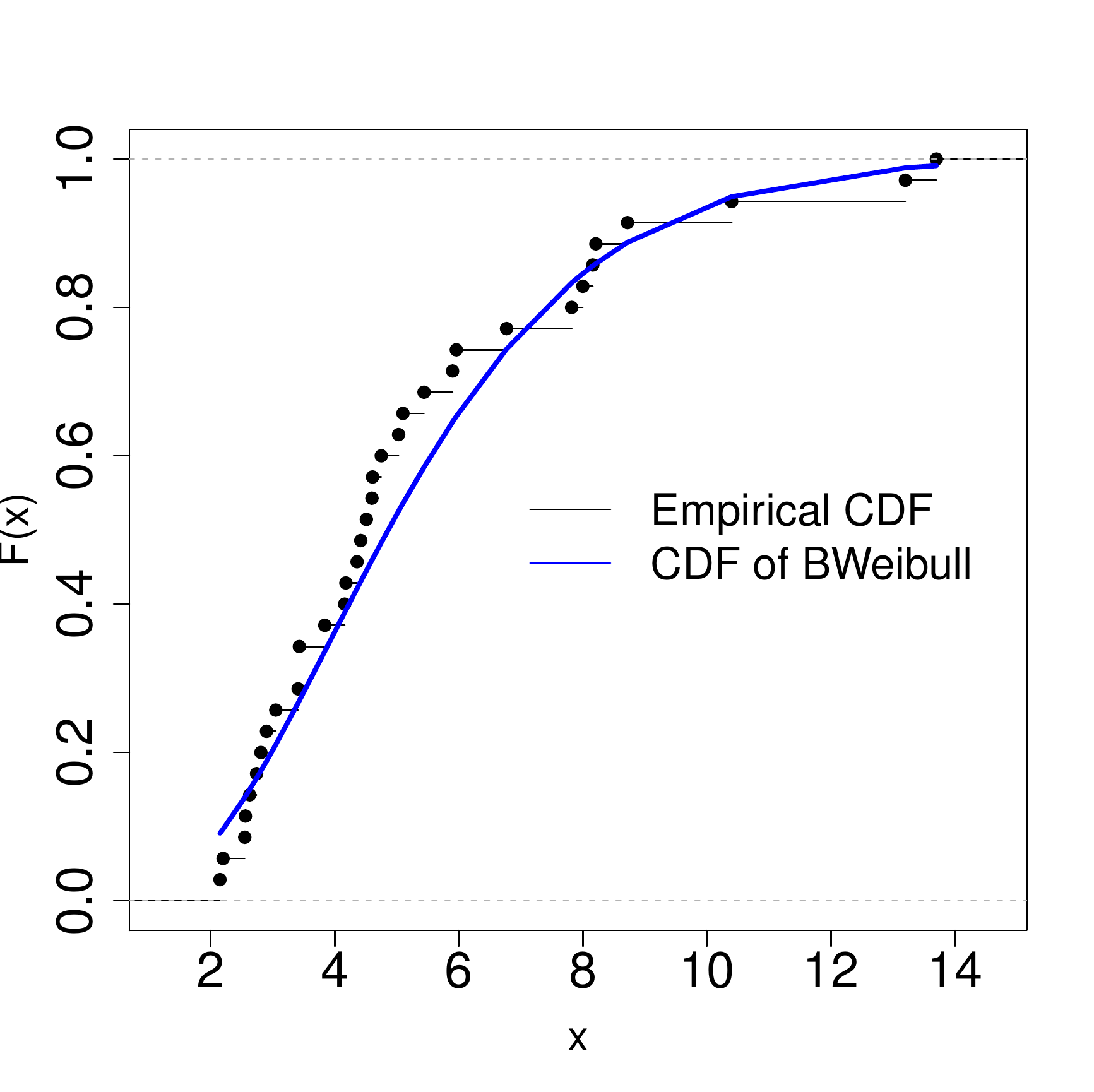}}
	\subfigure[Frequency and PDF of BWeibull($\widehat{{\theta}}_{\text{MLE}}$)]{\label{fig:hormonePDFq}\includegraphics[width=0.41\textwidth]{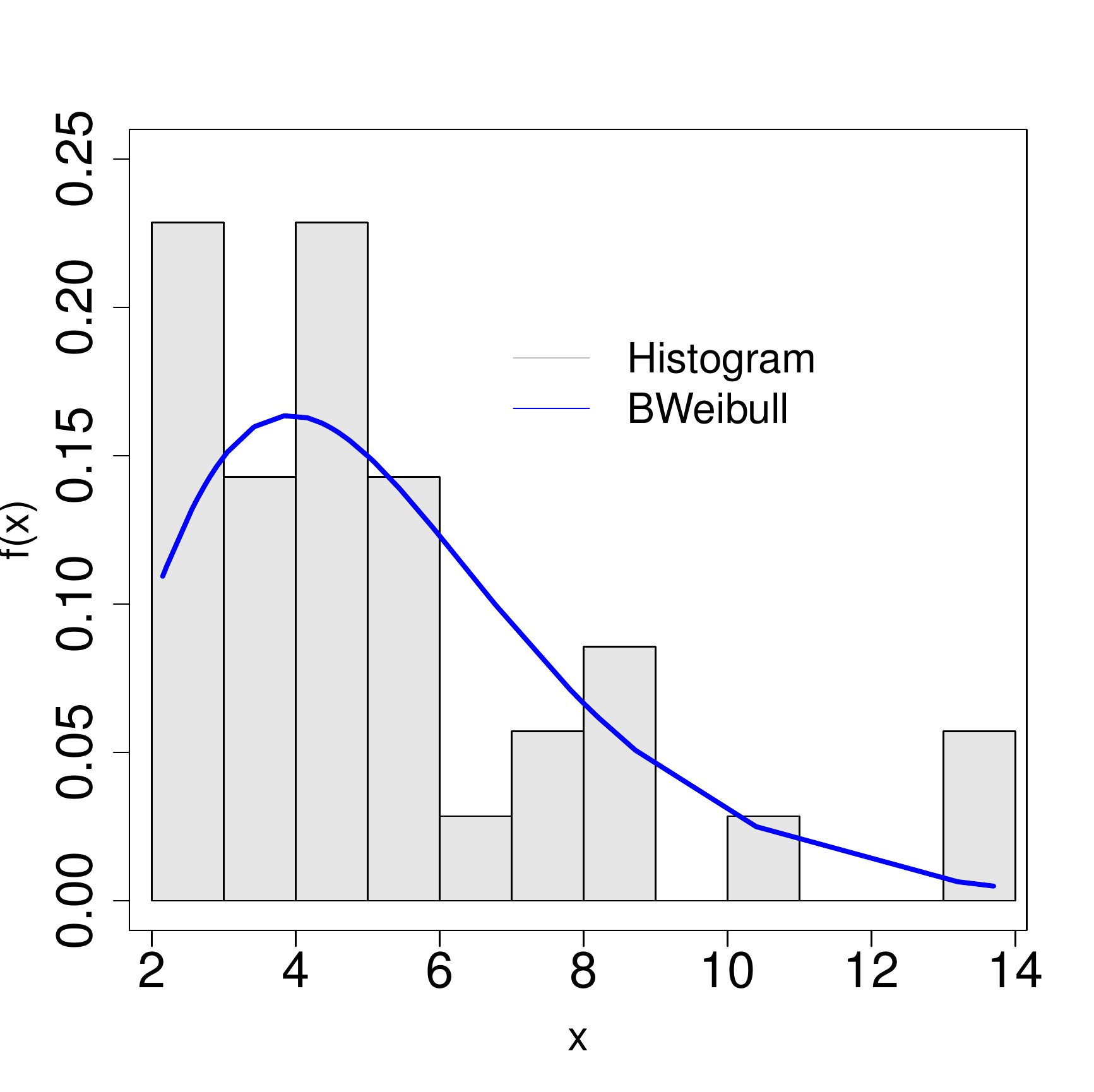}}
	\caption{The fitted values from MLE of parameters of bimodal Weibull distribution for  growth hormone data}
	\label{fighormone}
\end{figure}

{\bf Example 4 }: Wheaton River data are analyzed by \cite{Vila2020} and references therein. The number of observation is 72. MLqE with BWeibull give outperforming when we compare by MLqE with BGamma.  BGamma from \cite{Vila2020} has a good performance for fitting among the distributions \cite{Ortega11} and references therein. The advantage of using MLqE is observed when we look at the p-values of KS and CVM in Table \ref{tableriver}. 

\input{table4}
\begin{figure}[htbp]
	\centering
	\subfigure[(Empirical) CDF of BWeibull($\widehat{{\theta}}_{\text{MLqE}}$)]{\label{fig:riverCDFq}\includegraphics[width=0.41\textwidth]{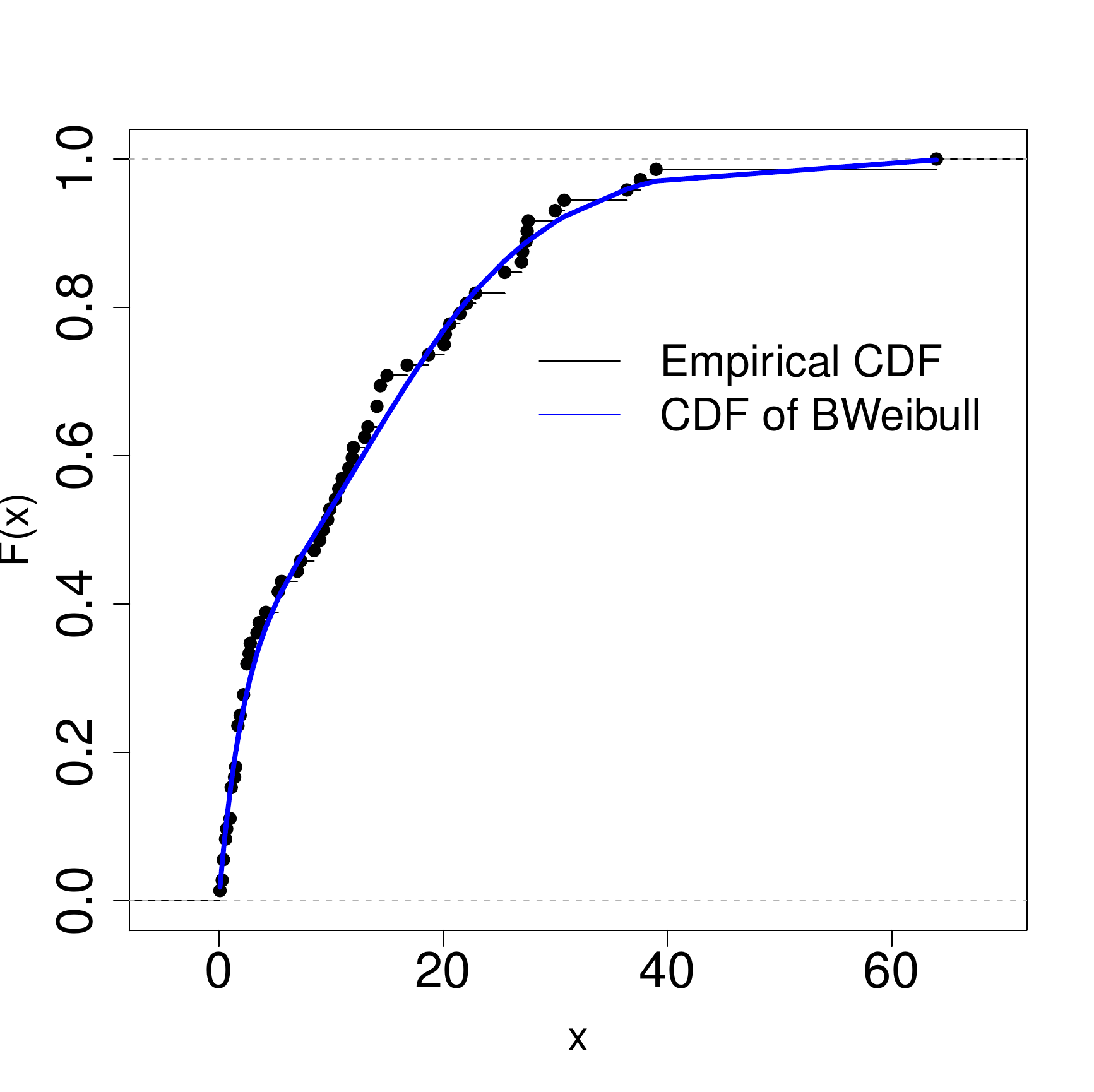}}
	\subfigure[Frequency and PDF of BWeibull($\widehat{{\theta}}_{\text{MLqE}}$)]{\label{fig:riverPDFq}\includegraphics[width=0.41\textwidth]{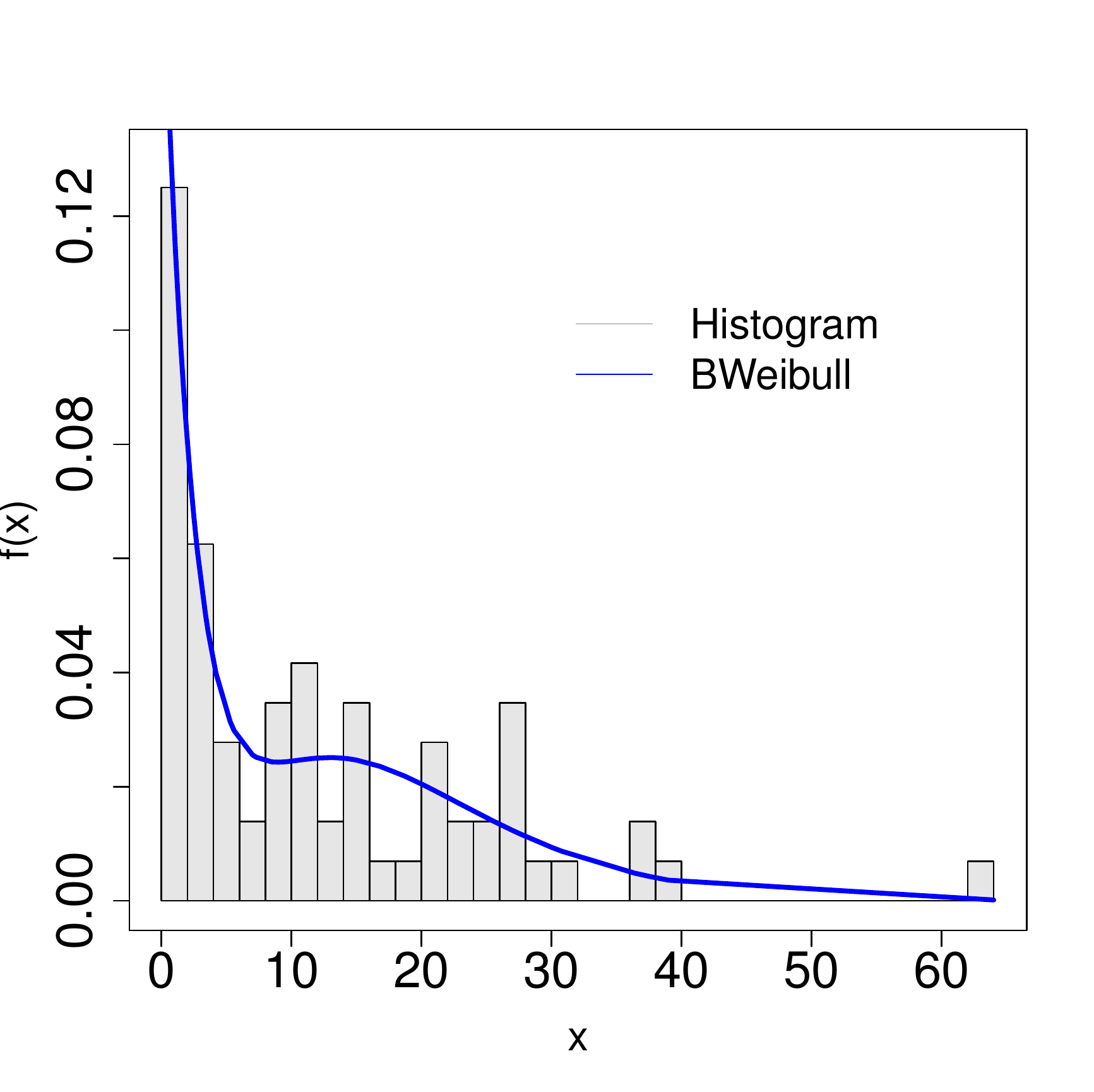}}
	\caption{The fitted values from MLqE of parameters of bimodal Weibull distribution for  Wheaton River data}
	\label{figriver}
\end{figure}
{\bf Example 5}: The death occurring due to cancer increases.  Gastric cancer data consisting of 125 observations are the censored data because of the difficulty of observering the patients who can die in the progress of medical treatment of patients. If data are censored, then there exists a replication for some observations, which leads to bimodaliy for frequency. For this reason, we prefer to use this data set which will be modelled by BWeibull and BGamma with $\log$ and $\log_q$ as objective functions. Let us note that censoring  designs \cite{censoring,Ngcensoring} can be originally regarded as an objective function \cite{Hampeletal86,MLqEBio,PhDthesis,CanKor18,Canorder20} used to fit a data set. Thus,  using $\log(f)$ and $\log_q(f)$ is reasonable and $f$ can be chosen as BWeibull and BGamma. Table \ref{tablecancer} shows that MLqE with BWeibull should be considered as good fitting performance assessed by p-value of KS even if p-value of CVM of BGamma with $\log$ as a little ahead from others is the best one among them.

\input{table5}
\begin{figure}[htbp]
	\centering
	\subfigure[(Empirical) CDF of BWeibull($\widehat{{\theta}}_{\text{MLqE}}$)]{\label{fig:gastricCDFq}\includegraphics[width=0.41\textwidth]{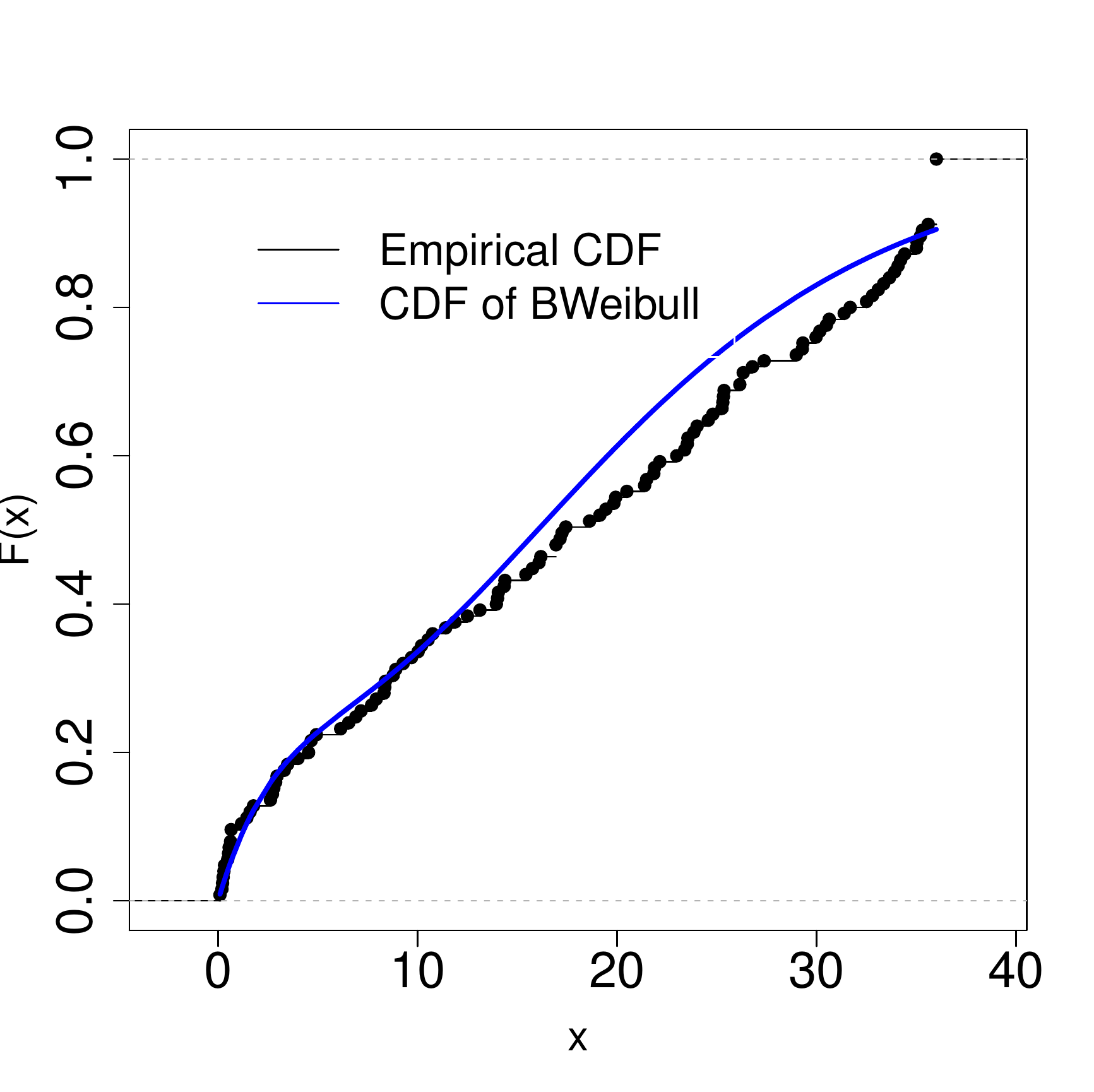}}
	\subfigure[Frequency and PDF of BWeibull($\widehat{{\theta}}_{\text{MLqE}}$)]{\label{fig:gastricPDFq}\includegraphics[width=0.41\textwidth]{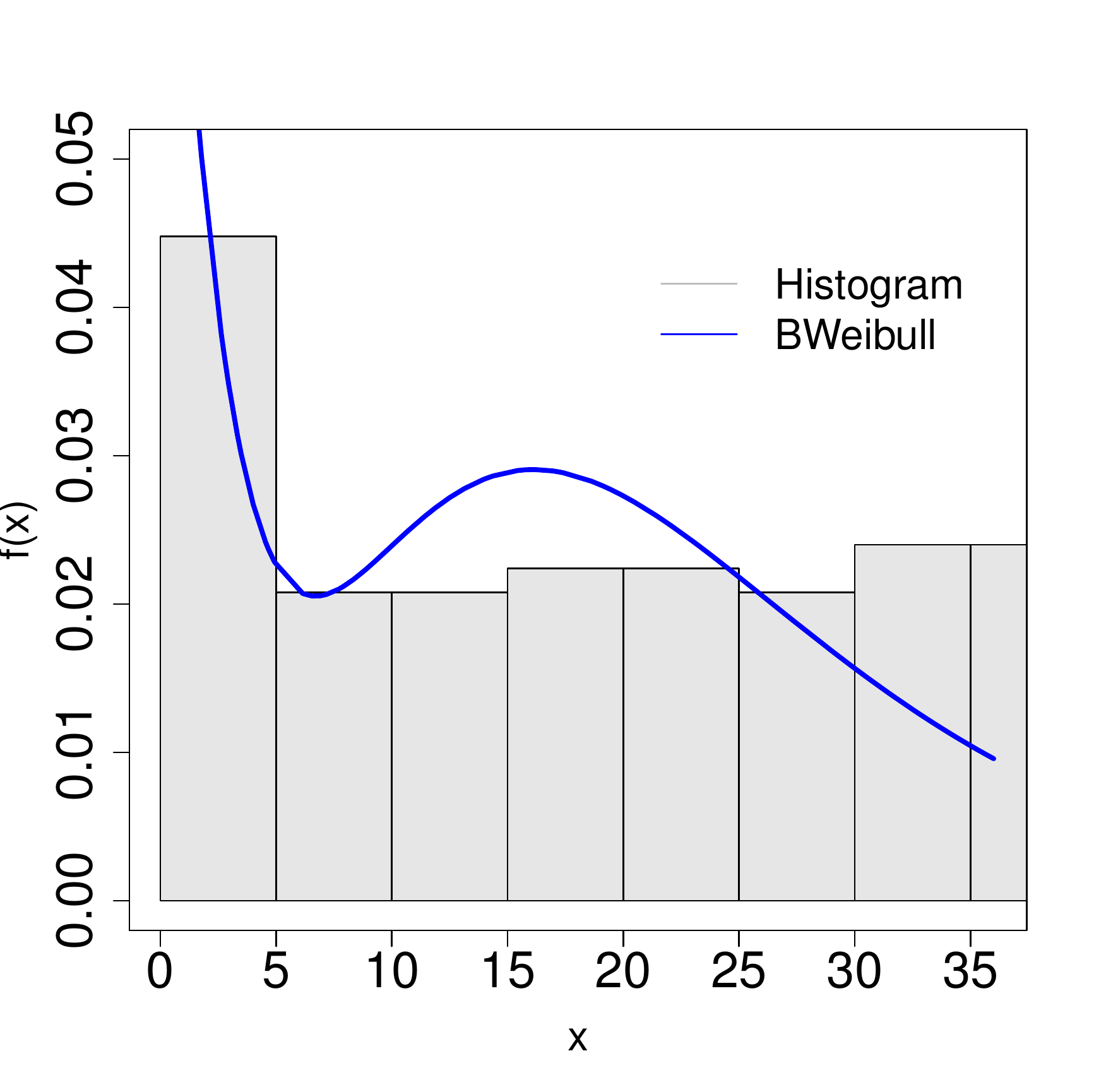}}
	\caption{The fitted values from MLqE of parameters of bimodal Weibull distribution for  gastric cancer data}
	\label{figgastric}
\end{figure}


\section*{Conclusions}\label{concsect}
\label{sec5}
A bimodal form of Weibull distribution has been derived. The properties such as unimodality, bimodality, moment, moment generation functions, entropies, etc. have been obtained. In order to get the estimators $\widehat{\boldsymbol{\theta}}$ of parameters, MLE and MLqE methods have been used. The concavity property of $\log$ and $\log_q$ are very important to apply for estimation of parameters \cite{Lindsay94}. Further, it should be noted that the existences of Shannon  and Tsallis $q$-entropies can also be important to apply MLE and MLqE methods. The estimates of $\widehat{\boldsymbol{\theta}}$ from MLE and MLqE are obtained by using of the heuristic algorithm because of nonlinearity of $\log(f)$ and $\log_q(f)$. The p.d. function $f$ has been chosen as  BWeibull and BGamma distributions. Instead of using a parametric model directly for modelling a data set via $\log$, it should be preferred to apply $\log_q$ which can derive the different forms of a parametric model for the different values of $q$ from MLqE method. Thus, we have obtained the results from numerical experiments showing that using MLqE method for parameters of BWeibull distribution is suggested. The anaytical expressions for the elements of FI based on $\log$ have been obtained. The square root of variance values of estimators $\widehat{\boldsymbol{\theta}}$ from MLE and MLqE are given in the numerical experiments. 

The application of the proposed distribution for case of the censored data will be studied and comparative study for different forms of the censoring designs \cite{censoring,Ngcensoring} will be applied to test performance of the censoring designs when we have the empirical distribution which has bimodality. The  properties such as existence, uniqueness of roots for $\log$ and $\log_q$ likelihood functions of the censoring designs which can be regarded as objective function \cite{Hampeletal86} will be studied.

\section*{Acknowledgments} 

This study was financed in part by the Coordenação de Aperfeiçoamento de Pessoal de Nível Superior - Brasil (CAPES) - Finance Code 001.
\appendix
\section*{Appendix: Main buildings of codes used for computation of parameters and goodness of fit test statistics}
The PDF is computed by
\begin{verbatim}
biwblpdf<- function(x,a,b,d){
	f=dweibull(x, a, b, log = FALSE);
	Z = 2 + ((2 *b *d *(-gamma(1/a) + b *d* gamma(2/a)))/a);
	g=(1 + (1 - d*x)^2)*f/Z;
}
\end{verbatim}
The optimization of $l_q$ in equation \eqref{likeliqtheta} is performed by following functions:
\begin{verbatim}
	logqLikFun <- function(p) {
		a <- p[1]; b <- p[2]; d <- p[3];
		-sum((biwblpdf(x,a,b,d)^(1-q)-1)/(1-q))
	}
	mlqe <- metaOpt(logqLikFun, optimType = "MIN", algorithm = "HS", 
	numVar,rangeVar,control = list(), seed = NULL)  
\end{verbatim}
The CDF arranged by manipulation and adopted for the computation of R platform is computed by
\begin{verbatim}
	CDFbiWei<- function(x,a,b,d){
		x=sort(x);   
		G=(-2 * b * d * gamma(1/a) + a * (2 - 2 * exp(-(x/b)^a) + 
		b^2 * d^2 * gamma(2/a+1) +   2 * b * d * (gamma(1 + 1/a) 
		* pgamma((x/b)^a, 1 + 1/a, 1, lower = FALSE)) 
		- b^2 * d^2 * (gamma(2/a+1) * pgamma((x/b)^a, 2/a+1, 1, lower = FALSE))))
		/(a * (2 + (2 * b * d * (-gamma(1/a) + b * d * gamma(2/a)))/a))
	}
\end{verbatim}
The values of goodness of fit test statistics are obtained by \cite{r1,r2}
\begin{verbatim}
	install.packages("CDFt") 
	library("CDFt") 
	val_CDFbiWeiq=CDFbiWei(sort(x),mleBiWeiabq[1],mleBiWeiabq[2],mleBiWeiabq[3])
	fun.ecdf=ecdf(x);my.ecdf <- fun.ecdf(sort(x));
	ks.test(my.ecdf, val_CDFbiWeiq,alternative = c("two.sided", "less", "greater"),
	exact = NULL, tol=1e-8,simulate.p.value=FALSE,B=2000)
	resq = CramerVonMisesTwoSamples(my.ecdf,val_CDFbiWeiq); 
	pvalueCVMq = 1/6*exp(-resq); 
\end{verbatim}

\end{document}

%% file: table1.tex
\begin{table}[htbp]
	 \caption{Inference values for parameters, KS and CVM values for  carbon fibers data}
	   \label{tablecarbonfibers}
				\scalebox{0.88}
	{   
	\begin{tabular}{c|ccccccc}
Model		& $\hat{\alpha}$ & $\hat{\beta}$ & $\hat{\delta}$ & KS & p-value(KS) & CVM & p-value(CVM)\\ \hline
BWeibull($\hat{{\theta}}_{\text{MLE}}$)		& 3.6961(0.0807)&  2.7482(0.0306)& 2.3073(1.2630) & 0.14&0.7112&0.0690 &0.1556 \\

BWeibull($\hat{{\theta}}_{\text{MLqE}}$), $q=0.8$		&  4.8707(0.1450)& 2.8434(0.0378)& 1.8986(2.4058) &0.06 & $\approx$ {\bf 1} &   0.0158& \bf{0.1641}\\

BGamma($\hat{{\theta}}_{\text{MLE}}$)		&14.9970(0.6521)& 5.8832(0.2449)& 1.5823(1.0677)   & 0.12& 0.8643 & 0.0734 &0.1549\\

BGamma($\hat{{\theta}}_{\text{MLqE}}$), $q=0.75$			& 14.9994(0.7332) &5.9217(0.3195) &1.7915(3.0210) &0.10 & 0.9639 & 0.0706 &0.1553 \\ \hline

	\end{tabular}
}
\begin{footnotesize}
KS: Kolmogorov-Smirnov, CVM: Cram\'{e}r–von Mises, {\bf Bold} represents the best fitting. 
\end{footnotesize}
\end{table}


%% file: table2.tex
\begin{table}[htbp]
		\caption{Inference values for parameters, KS and CVM values for o3max data }
		  \label{tableo3max}
			\scalebox{0.86}
	{
	\begin{tabular}{c|ccccccc}
Model		& $\hat{\alpha}$ & $\hat{\beta}$ & $\hat{\delta}$ & KS & p-value(KS) & CVM & p-value(CVM)\\ \hline
BWeibull($\hat{{\theta}}_{\text{MLE}}$)	&10.0851(0.7039)& 0.1728(0.0018)& 14.9983(10.9941) & 0.1667 &{\bf 0.8928} 
 & 0.0556  &{\bf 0.1577}  \\
BWeibull($\hat{{\theta}}_{\text{MLqE}}$), $q=0.99$		&  10.1609(0.7079)& 0.1728(0.0018) & 14.9995(11.1710) & 0.1667 &{\bf 0.8928} 
&0.0556 & {\bf 0.1577}\\
BGamma($\hat{{\theta}}_{\text{MLE}}$)		&3.2524(0.2537) &14.9963(1.0863)& 1.6134(1.0448) & 0.4583 & 0.0129 &0.6667 & 0.0856
 \\
BGamma($\hat{{\theta}}_{\text{MLqE}}$), $q=0.99$	&  3.2513(0.2542) &14.9935(1.0867) &1.6130(1.0585)& 0.4583 & 0.0129 &0.6667 & 0.0856
 \\ \hline
	\end{tabular}
}
\end{table}


%% file: table3.tex
\begin{table}[htbp]
	\caption{ Inference values for parameters, KS and CVM values for  growth hormone data}
	\label{tableghormone}
	\scalebox{0.88}
	{
		\begin{tabular}{c|ccccccc}
			Model		& $\hat{\alpha}$ & $\hat{\beta}$ & $\hat{\delta}$ & KS & p-value(KS) & CVM & p-value(CVM)\\ \hline
			BWeibull($\hat{{\theta}}_{\text{MLE}}$)	& 1.1344(0.0370) &2.1303(0.1213)& 3.8210(1.0365)	 &0.1429 &  0.8674
			& {\it 0.0651} & {\it0.1562}\\
			BWeibull($\hat{{\theta}}_{\text{MLqE}}$), $q=0.85$		& 1.2489(0.0548) &2.2989(0.1541)& 3.5132(1.4327) &0.0857 &0.9995 & 0.0357 &0.1608 \\
			BGamma($\hat{{\theta}}_{\text{MLE}}$)&1.9136(0.1685) &0.7620(0.0319) &3.6132(1.7301) &0.1429 &0.8674 & {\it 0.0765} &{\it 0.1544} \\
			BGamma($\hat{{\theta}}_{\text{MLqE}}$), $q=0.87$&2.4648(0.2891) &0.9224(0.0513) &3.2896(2.9941) &0.0857 &0.9995 &0.0349 & {\bf 0.1610} \\ \hline
		\end{tabular}
	}
	\begin{footnotesize}
		{\it Italic} represents  comparison for the best fitting for MLE of BWeibull and BGamma. 
	\end{footnotesize}
\end{table}


%% file: table4.tex
\begin{table}[htbp]
	\caption{Inference values for parameters, KS and CVM values for  Wheaton River data}
	  \label{tableriver}
			\scalebox{0.88}
	{
	\begin{tabular}{c|ccccccc}
Model		& $\hat{\alpha}$ & $\hat{\beta}$ & $\hat{\delta}$ & KS & p-value(KS) & CVM & p-value(CVM)\\ \hline
BWeibull($\hat{{\theta}}_{\text{MLE}}$)	& 0.9721(0.0102) &5.4258(0.1397)& 0.1924(0.0050) & 0.0694 & 0.9951&0.0260 &  0.1624\\
BWeibull($\hat{{\theta}}_{\text{MLqE}}$), $q=0.99$	&  0.9770(0.0103)& 5.4536(0.1411)& 0.1910(0.0050) & 0.0694 &{\bf 0.9951} &0.0206 &{\bf 0.1633} \\
BGamma($\hat{{\theta}}_{\text{MLE}}$)	&  1.0591(0.0182)& 0.1768(0.0023)& 0.1776(0.0039) &0.0833 &0.9639
 &0.0336 & 0.1612 \\
		BGamma($\hat{{\theta}}_{\text{MLqE}}$), $q=0.99$&  1.0568(0.0184) &0.1780(0.0024)& 0.1783(0.0040) &0.0833 &0.9639
		 & 0.0287&0.1620 \\\hline
	\end{tabular}
}
\end{table}


%% file: table5.tex
\begin{table}[htbp]
	\caption{Inference values for parameters, KS and CVM values for  gastric cancer data}
	  \label{tablecancer}
		\scalebox{0.88}
	{
	\begin{tabular}{c|ccccccc}
		Model		& $\hat{\alpha}$ & $\hat{\beta}$ & $\hat{\delta}$ & KS & p-value(KS) & CVM & p-value(CVM)\\ \hline
BWeibull($\hat{{\theta}}_{\text{MLE}}$)	&  1.1181(0.0092)& 8.4517(0.1373)& 0.1764(0.0031) & 0.104 & 0.5085& 0.0962&0.1514 \\
BWeibull($\hat{{\theta}}_{\text{MLqE}}$), $q=0.9$&1.0092(0.0097) &6.7659(0.1403)& 0.2106(0.0043) &0.096 &{\bf 0.6121} &0.1390 & 0.1450 \\
BGamma($\hat{{\theta}}_{\text{MLE}}$)& 0.9531(0.0133) &0.1433(0.0030)& 0.2193(0.0031) &0.112 & 0.4131& 0.0849 &{\bf 0.1531} \\
BGamma($\hat{{\theta}}_{\text{MLqE}}$), $q=0.9$& 0.9439(0.0147) &0.1481(0.0012)& 0.2225(0.0036) &0.104 &  0.5085&  0.1422& 0.1446 \\ \hline
	\end{tabular}
}
\end{table}
